\documentclass[12pt]{article}
\usepackage{amssymb,amsmath}
\usepackage{amsthm}
\usepackage{xcolor}
\usepackage[latin1]{inputenc}
\usepackage{url}
\usepackage{cancel}
\usepackage{comment}
\newtheorem{theorem}{Theorem}
\newtheorem{proposition}[theorem]{Proposition}
\newtheorem{lemma}[theorem]{Lemma}

\theoremstyle{definition}
\newtheorem{definition}[theorem]{Definition}

\newtheorem{remark}[theorem]{Remark}
\def\d{\partial}
\def\f{\frac}

\def\A{\mathcal{A}}

\newcommand*{\pd}
[2]{\mathchoice{\frac{\partial#1}{\partial#2}}
  {\partial#1/\partial#2}{\partial#1/\partial#2}
  {\partial#1/\partial#2}}

\newcommand*{\fd}
[2]{\mathchoice{\frac{\delta#1}{\delta#2}}
  {\delta #1/\delta#2}{\delta#1/\delta#2}{\delta#1/\delta#2}}
\newcommand{\ddx}[1]{\partial_x^{#1}}
\allowdisplaybreaks[3]

\begin{document}

\title{Three computational approaches \\
  to weakly nonlocal Poisson brackets} \author{M. Casati$^1$, P. Lorenzoni$^2$,
  R. Vitolo$^3$} \date{}
\maketitle
\vspace{-7mm}
\begin{center}
  $^1$School of Mathematics, Statistics and Actuarial Science
  \\
  University of Kent, Canterbury, UK
  \\
  email: \texttt{M.Casati@kent.ac.uk}
  \\
  $^2$Dipartimento di Matematica
  \\
  Universit\`a di Milano-Bicocca, Milano, Italy
  \\
  email: \texttt{paolo.lorenzoni@unimib.it}
  \\
  $^3$Dipartimento di Matematica e Fisica ``E. De Giorgi'',
  \\
  Universit\`a del Salento and Sezione INFN di Lecce
  \\
  via per Arnesano, 73100 Lecce, Italy
  \\
  email: \texttt{raffaele.vitolo@unisalento.it}
\end{center}

\begin{abstract}
  We compare three different ways of checking the Jacobi identity for weakly
  nonlocal Poisson brackets using the theory of distributions, of
  pseudodifferential operators and of Poisson vertex algebras, respectively. We
  show that the three approaches lead to similar computations and same results.
\end{abstract}

\tableofcontents

\section{Introduction}

Poisson brackets play a fundamental role in Mathematical Physics. Many partial
differential equations (PDEs) can be endowed by Poisson brackets by means of a
Hamiltonian operator; the prototype of such equations is the Korteweg--de Vries
(KdV) equation, that was shown to be a completely integrable Hamiltonian system
in \cite{ZF71,G71}. More precisely, an autonomous system of PDEs in
evolutionary form
\begin{equation}
\label{eq:4}
F = u^i_t - f^i(u^j,u^j_\sigma) = 0
\end{equation}
in two independent variables $t$, $x$ and in $n$ dependent variables $u^1$,
\dots, $u^n$, where $u^j_\sigma$ are $\sigma$-times $x$-derivatives of $u^j$,
is said to be \emph{Hamiltonian} if it can be written as
\begin{equation}
\label{eq:6}
u^i_t = f^i(u^j,u^j_\sigma) = A^{ij}\fd{H}{u^i},
\end{equation}
where $\fd{}{u^i}$ are the variational derivatives,
$H=\int h(u^j,u^j_\sigma)\,dx$ is the Hamiltonian functional and
$A^{ij}=A^{ij\sigma}\partial_\sigma$ is a matrix partial differential operator
in total derivatives $\partial_\sigma=\partial_x\circ\cdots\circ\partial_x$
($\sigma$-times). The operator $A$ is required to define a Poisson bracket on
the space of local functionals as follows. Given two local functionals
$F=\int f\,dx$ and $G=\int g\,dx$ where $f=f(u^j,u^j_\sigma)$ and similarly for
$g$, the operation
\begin{equation}
\label{eq:7}
\{F,G\}_P = \int P^{ij\sigma}\fd{f}{u^i}\partial_\sigma\fd{g}{u^j}\,dx
\end{equation}
is required to be a Poisson bracket. This requirement implies strong
differential constraints on the coefficients
$P^{ij\sigma}= P^{ij\sigma}(u^k,u^k_\sigma)$ coming from the skew-symmetry and
the Jacobi identity for $\{\,,\,\}_P$, as we will see later on.

It was soon realized that Poisson brackets for integrable systems should be
extended to include nonlocal terms (or pseudo-differential operators) in the
definition: one of the first examples is the Poisson bracket defined by the
operator $u_x\partial_x^{-1}u_x$ for the Krichever--Novikov equation appeared
in \cite{Sok84}. A general class of nonlocal Poisson brackets was introduced
and studied by Ferapontov and Mokhov in \cite{FM90} and further generalized in
\cite{Fer91,Fer91b}. This class is made of operators $P$ of the form
\begin{equation}
\label{eq:8}
P^{ij} = g^{ij}\partial_x + \Gamma^{ij}_ku^k_x +
c^\alpha w^i_{\alpha k}u^k_x \partial_x^{-1} w^j_{\alpha h}u^h_x,
\end{equation}
where all coefficients are functions of the field variables $(u^i)$ and $g$ is
assumed to be non degenerate.  The operator $P$ is homogeneous with respect to
$x$-derivation in the sense of Dubrovin and Novikov \cite{DN}, and generalizes
the local homogeneous Poisson brackets defined therein. Purely nonlocal
Hamiltonian operators generalizing the nonlocal structure of Krichever-Novikov
equation have been studied in \cite{M,GLR}.

The above class was further enlarged in \cite{MN01} to Poisson brackets defined
by operators of the form\footnote{We will use the form of the nonlocal part
	that is called ``reduced'' in \cite{MN01}.}
\begin{equation}
\label{eq:10}
P^{ij} = B^{ij\sigma}\partial_\sigma  +
e^{\alpha}w^i_{\alpha} \partial_x^{-1} w^j_{\alpha},
\end{equation}
where the coefficients $B^{ij\sigma}$ and $w^i_\alpha$, $w^j_\beta$ can depend
on all field variables $u^k$ and a finite number of derivatives $u^k_\sigma$,
and $e^{\alpha}$ are constants. The operator
$\partial_x^{-1}$ is defined to be
\begin{equation}
\label{eq:11}
\partial_x^{-1} = \frac{1}{2}\int_{-\infty}^x\,dx -
\frac{1}{2}\int^{+\infty}_x\,dx
\end{equation}
acting on a suitable space of rapidly decreasing vector functions; the operator
is formally skew-adjoint. Poisson brackets defined by operators of the type of
$J$ are said to be \emph{weakly non-local} \cite{MN01}; Poisson brackets of the
subclass of $P$ in \eqref{eq:8} are said to be weakly non-local Poisson
brackets of hydrodynamic type (PBHT).

Weakly non-local Poisson brackets play an important role in the theory of
integrable systems. It was proved in \cite{MN01} that the higher Hamiltonian
operators for the KdV equation generated by composition with the Recursion
operator (known since the late '70) are all weakly non-local. It was
conjectured that every hydrodynamic system of PDEs
\begin{equation}
\label{eq:9}
u^i_t = V^i_j(u^k) u^j_x
\end{equation}
is Hamiltonian with respect to a suitable weakly nonlocal Poisson bracket of
hydrodynamic type. As the dispersionless limit of a large class of evolutionary
systems of PDEs consists in a system of the type of~\eqref{eq:9}, it is natural
to expect that weakly nonlocal Poisson brackets, and in particular those that
can be obtained as a deformation of a weakly nonlocal Poisson bracket of
hydrodynamic type, are of utmost interest.

Despite their importance, the use of nonlocal Poisson brackets is quite
limited, especially if compared to local homogeneous Poisson brackets of
hydrodynamic type. The most important reason for that is related to the much
higher computational difficulties with respect to the local case.
Indeed, the very first moment when such difficulties are met
is the verification of the Jacobi identity for the coefficients of a candidate
differential operator $P$.

More in detail, the problem is to reach a divergence-free form of the Jacobi
expression $\{\{F,G\}_P,H\}_P + \text{cyclic}$ when there are nonlocal terms,
in order to check (or impose) the vanishing of a necessary and sufficient set
of differential expressions of the coefficients of $P$.

In the literature we can find (at least) three approaches to the Hamiltonian
formalism for PDEs:
\begin{enumerate}
\item the approach with distributions, quite spread in Theoretical Physics
  \cite{DN,DN2,DZ};
\item the approach with differential operators, more common between pure
  mathematicians \cite{Many,Dorf,Olv93};
\item a new algebraic approach based on Poisson Vertex Algebras, introduced in
  \cite{BDSK09} for local Poisson brackets and later extended to nonlocal
  Poisson brackets in \cite{DSK13}.
\end{enumerate}
The third approach relies on an algebraic formulation in order to give a
mathematically well-grounded basis to the theory of nonlocal Poisson
brackets. In particular, a canonical form of pseudo-differential operators is
achieved by the systematic use of formal series (see \cite{Olv93}), which, on
the other hand, lay outside the more traditional differential-geometric
picture (see the Conclusions for further remarks).

The aim of this work is to illustrate an \emph{algorithmic procedure to compute
	the Jacobi identity for weakly nonlocal Poisson brackets} in the three
formalisms above.

In the case of distributions, the algorithm has been introduced in~\cite{L} in
order to study the bi-Hamiltonian structure of the Liouville and sine-Gordon
PDEs. In the case of differential operators, the algorithm is shown here for
the first time thanks to the explicit correspondence between the languages of
distributions and differential operators. In the case of Poisson Vertex
Algebra, the algorithm has been described in \cite{DSK13} (a nontrivial example
can be found in~\cite{CFPV}).

In all cases, the algorithm consists in the reduction of the Jacobi identity to
a reduced form, in such a way to make the identity effectively computable. This
is practically achieved by means of identities between distributions, or
integration by parts, or algebraic manipulations. In all cases the interesting
fact is that the computations lead to the same conditions in the end, with a
precise correspondence that is explicitly written as a `dictionary' between the
three formalisms.

In order to illustrate the algorithm and the correspondence between the
different formalisms we will consider the case of weakly nonlocal Poisson
brackets of hydrodynamic type. We will obtain the following conditions on the
coefficients in $P^{ij}$ (assuming $\det g^{ij}\neq 0$), namely
\begin{subequations}\label{eq:61}
	\begin{gather}
	\label{eq:52}
	g^{ij} = g^{ji},
	\\
	\label{eq:56} g^{ij}_{,k} =
	\Gamma^{ij}_k+\Gamma^{ji}_k,
	\\
	\label{eq:57} g^{is}\Gamma^{jk}_s = g^{js}\Gamma^{ik}_s,
	\\
	\label{eq:58} g^{is}w^j_s = g^{js}w^i_s
	\\
	\label{eq:59} \nabla_i w^j_k = \nabla_k w^j_i,
	\\ \label{eq:60} R^{ij}_{kh} = w^i_kw^j_h - w^j_kw^i_h.
	\end{gather}
\end{subequations}
where $\nabla$ is the linear connection with Christoffel symbols
$\Gamma^k_{ij}=-g_{il}\Gamma^{lk}_j$ and $R^{ij}_{kh}=g^{is}R^j_{skh}$ is the
Riemannian curvature. The above conditions, first obtained in \cite{Fer91} (for
details of computations see \cite{Ont}), admit the following interpretation:
the first three equations tell us that the functions $g^{ij}$ can be
interpreted as the contravariant components of a (pseudo)-euclidean metric and
$\Gamma^{ij}_k$ as Christoffel symbols of the corresponding Levi-Civita
connection, while the remaining equations coincide with the classical
Gauss-Peterson-Mainardi-Codazzi equations for a hypersurface in a
(pseudo)-Euclidean space. The (pseudo)-metric $g$ and the affinor $w$ can be
identified with the induced metric and the Weingarten operator respectively.

\bigskip

The paper is organized as follows. In Sections $2,3$ and $4$ we explain the
algorithm to check Jacobi identity in the three formalisms and we write a sort
of dictionary between the three approaches.  The remaining sections are devoted
to illustrate the algorithm in the case of weakly nonlocal Poisson brackets of
hydrodynamic type.  In order to simplify the computations we consider the
subcase where the nonlocal tail contains a single term but the computations can
be performed in the same way in the general case.

\paragraph{Acknowledgments.} We thank E. Ferapontov, J. Krasil'shchik,
M. Pavlov, D. Valeri, A. Verbovetsky, Y. Zhang for useful discussions. M.~C. is
supported by EPSRC grant EP/P012698/1. P.L. is supported by MIUR - FFABR funds
2017 and by research funds of the University of Milano - Bicocca. P.L is
grateful to the Department of Mathematics and Physics ``E. De Giorgi'' of the
Universit\`a del Salento for the kind hospitality and for supporting his visit.
R.V. acknowledges the support of the Department of Mathematics and Physics
``E. De Giorgi'' of the Universit\`a del Salento and of the Istituto Nazionale
di Fisica Nucleare, Sezione di Lecce, IS-CSN4 Mathematical Methods of Nonlinear
Physics.  Finally, we thank GNFM for supporting activities that contributed to
the research reported in this paper.

\section{Jacobi identity and distributions}
\label{sec:jacobi-ident-distr}

Here we briefly introduce weakly nonlocal Poisson brackets as distributions and
describe the algorithm for bringing the Jacobi identity to a reduced canonical
form.

\subsection{The Jacobi identity}
\label{sec:jacobi-identity}

Following \cite{MN01}, we consider weakly nonlocal Poisson brackets of the form
\begin{multline}\label{eq:55}
  \{u^i(x),u^j(y)\}_P=
  \sum_{k\ge 0}B^{ij}_k(u^h,u^h_\sigma)\delta^{(k)}(x-y)
  \\
  + e^\alpha w^i_{\alpha}(u^k,u^k_\sigma)
  \nu(x-y)w^j_{\alpha}(u^k,u^k_\sigma)
\end{multline} 
where $\nu(x-y)=\f{1}{2}\operatorname{sgn}(x-y)$.

The Jacobi identity
\begin{multline}\label{jacobi}
  \{u^i(x),u^j(y)\}_P,u^k(z)\}_P+\{u^k(z),u^i(x)\}_P,u^j(y)\}_P
  \\ +\{u^j(y),u^k(z)\}_P,u^i(x)\}_P=0
\end{multline}
can be written as \cite{DZ}
\begin{multline}
  \label{SBDZ}
J^{ijk}_{xyz}=
\frac{\partial P^{ij}_{x,y}}{\partial u^{l}_\sigma(x)}
\partial_x^\sigma P^{lk}_{x,z}
+ \frac{\partial P^{ij}_{x,y}}{\partial u^{l}_\sigma(y)}
\partial_y^\sigma P^{lk}_{y,z}
+\frac{\partial P^{ki}_{z,x}}{\partial u^{l}_\sigma(z)} \partial_z^\sigma
P^{lj}_{z,y}+
\\
\frac{\partial P^{ki}_{z,x}}{\partial u^{l}_\sigma(x)}
\partial_x^\sigma P^{lj}_{x,y}
+\frac{\partial P^{jk}_{y,z}}{\partial u^{l}_\sigma(y)}
\partial_y^\sigma P^{li}_{y,x}
+ \frac{\partial P^{jk}_{y,z}}{\partial u^{l}_\sigma(z)}
\partial_z^\sigma P^{li}_{z,x}=0,
\end{multline}
where $P^{ij}_{x,y}=\{u^i(x),u^j(y)\}_P$. The vanishing of the distribution
$J^{ijk}_{xyz}$ means that for any choice of the test functions
$p_i(x),q_j(y),r_k(z)$ the triple integral
\begin{equation}
  \iiint  J^{ijk}_{xyz} p_i(x)q_j(y)r_k(z)\,dxdydz
\end{equation}
should vanish.

\subsection{The algorithm}
\label{sec:algorithm}

Following \cite{L}, we present a procedure to collect together all terms which
are related by a distributional identity. We call the result of this procedure
the \emph{reduced form} of the Jacobi identity.
\begin{enumerate}
\item Using the identity
\begin{equation}\label{eq:3}
\nu(z-y)\delta(z-x)=\nu(x-y)\delta(x-z)
\end{equation}
and its two obvious analogues obtained by a cyclic permutation of the
variables, together with their differential consequences, we can eliminate all
terms containing $\nu(z-y)\delta^{(n)}(z-x)$, $\nu(y-x)\delta^{(n)}(y-z)$,
$\nu(x-z)\delta^{(n)}(x-y)$ producing nonlocal terms containing
$\nu(x-y)\delta^{(n)}(x-z)$, $\nu(z-x)\delta^{(n)}(z-y)$,
$\nu(y-z)\delta^{(n)}(y-x)$ and additional local terms.
\item  Using the identity
\begin{equation}
\label{id1}
f(z)\delta^{(n)}(x-z)=\sum_{k=0}^n\binom{n}{k}f^{(n-k)}(x)\delta^{(n-k)}(x-z),
\end{equation}
(and its cyclic permutations)
we can eliminate the dependence on $z$ in the coefficients of the terms
containing $\nu(x-y)\delta^{(n)}(x-z)$, the dependence on $y$ in the
coefficients of the terms containing $\nu(z-x)\delta^{(n)}(z-y)$ and the
dependence on $x$ in the coefficients of the terms containing
$\nu(y-z)\delta^{(n)}(y-x)$.  After the first two steps the nonlocal part of
$J^{ijk}_{xyz}$ has the form
\begin{multline}\label{eq:13}
a_1(x,y,z)\nu(x-y)\nu(x-z) + \text{cyclic}(x,y,z)\\
+ \sum_{n\ge 0}b_n(x,y)\nu(x-y)\delta^{(n)}(x-z)+\text{cyclic}(x,y,z).
\end{multline} 
\item The local part of $J^{ijk}_{xyz}$ (which contain also some additional
  terms coming from the nonlocal part) can be treated as usual and reduced to
  the form
\begin{equation}\label{eq:15}
  \sum_{m,n}e_{mn}(x)\delta^{(m)}(x-y)\delta^{(n)}(x-z)
\end{equation} 
using the identities (and their differential consequences)
\begin{equation}
  \delta(z-x)\delta(z-y)=\delta(y-x)\delta(y-z)=\delta(x-y)\delta(x-z)
  \label{eq:14}
\end{equation}

and the identities~\eqref{id1}.
\end{enumerate}
It is easy to check that no further simplifications are possible. We will see
later that the fulfillment of the Jacobi identity turns out to be equivalent to
the vanishing of each coefficient in the reduced form.

\section{Jacobi identity and differential operators}
\label{sec:jacobi-ident-diff}

\subsection{The Jacobi identity}
\label{sec:jacobi-identity-1}

The conditions under which the bracket~\eqref{eq:7} is a Poisson bracket can be
written as requirements on the differential operator $P$~\eqref{eq:10}. We
recall that the operator $P$ is a variational bivector \cite{Many,Dorf,Olv93},
hence it is defined up to total divergencies. We consider Poisson brackets
defined by differential operators of the form
$P^{ij} = B^{ij\sigma}\partial_\sigma  +
e^{\alpha}w^i_{\alpha} \partial_x^{-1} w^j_{\alpha}$~\eqref{eq:10}.
Then, it is well-known that
\begin{itemize}
\item the skew-symmetry of $\{,\}_J$ is equivalent to the formal
  skew-adjointness of $P$, $P^*= - P$;
\item the Jacobi identity for $\{,\}_P$ is equivalent to the vanishing of the
  \emph{Schouten bracket} $[P,P]=0$.
\end{itemize}
Note that the Schouten bracket of two variational bivectors is a variational
three-vector, \emph{i.e.}, it is a skew-symmetric differential operator with
three arguments whose value is defined up to total divergencies.

In coordinates, the formal adjoint $P^*$ of the operator $P$ is
\begin{equation}
  \label{eq:54}
  P^* (\psi)^j = (-1)^{|\sigma|}\partial_\sigma(B^{ij\sigma}\psi_i) -
  e^\alpha w^j_{\alpha} \partial_x^{-1}(w^i_{\alpha}\psi_i).
\end{equation}
We stress that the non-local summand of weakly nonlocal operators is
skew-adjoint by construction: we have $(\partial_x^{-1})^*= - \partial_x^{-1}$.

Let us denote by $\ell_{P,\psi}(\varphi)$ the linearization of the
(coefficients of the) operator $P$. We have the following coordinate
expressions:
\begin{multline}
  \label{eq:17}
  \ell_{P,\psi}(\varphi)^i = \pd{B^{ij\sigma}}{u^k_\tau}\partial_\sigma\psi^1_j
  \partial_\tau \varphi^k +
  e^\alpha \pd{w^{i}_\alpha}{u^k_\tau}\partial_\tau \varphi^k
  \partial_x^{-1}(w^j_\alpha\psi_j)
  \\
  + e^\alpha w^{i}_\alpha\partial_x^{-1}
  \left(
    \pd{w^j_\alpha}{u^k_\tau}\partial_\tau \varphi^k\psi_j
  \right),
\end{multline}
where we used \eqref{eq:25} and the fact that $\partial_x^{-1}$ commutes with
linearization. Then, we have the following expression for the Schouten bracket:
\begin{equation}
  \label{eq:283}
  [P,P](\psi^1,\psi^2,\psi^3) = 2[\ell_{P,\psi^1}(P(\psi^2))(\psi^3)
    + \text{cyclic}(\psi^1,\psi^2,\psi^3)],
\end{equation}
where square brackets denote the fact that the expression is calculated up to
total divergencies.
We observe that the expression of the Schouten bracket of two operators can be
written in different ways, which differ up to total divergencies. In the
Appendix we wrote two more expressions that are more commonly used in the
formalism of differential operators, together with a proof of their
equivalence.

\subsection{Dictionary: distributions and differential operators}
\label{sec:dict-distr-diff}

Here we present a dictionary between the language of operators and the language
of distributions for the reader's convenience. The calculus with distributions
is defined in \cite[Subsect.\ 2.3]{DZ}.

The following notation for a local multivector coincide:
\begin{align}
  \label{eq:23}
  P = & B^{i_1}{}^{i_2}_{\sigma_2}{}^{\cdots}_{\cdots}{}^{i_k}_{\sigma_k}
        (u^i(x^1),u^i_\sigma(x^1)) \,
        \delta^{(\sigma_2)}(x^1-x^2)\cdots\delta^{(\sigma_k)}(x^1-x^k),
  \\\label{eq:23bis}
  P = & \int B^{i_1}{}^{i_2}_{\sigma_2}{}^{\cdots}_{\cdots}{}^{i_k}_{\sigma_k}
        \psi^1_{i_1}\partial^{\sigma_2}_{i_2}\psi^2\cdots
        \partial^{\sigma_k}_{i_k}\psi^k\, dx.
\end{align}
In particular the value of the multivector in the distributional notation is
obtained by evaluating it on test vector functions of the arguments $x^1$,
\dots, $x^k$. The above correspondence can easily be extended between the
nonlocal multivectors~\eqref{eq:55} and~\eqref{eq:10}. Then, it is clear that
the expressions~\eqref{SBDZ} and~\eqref{eq:283} coincide up to the evaluation
on test vector functions.

\subsection{The algorithm}
\label{sec:algorithm-2}

The result of the Schouten bracket $[P,P]$~\eqref{eq:283} is a three-vector
 and has the following coordinate expression:
\begin{equation}
  \label{eq:19}
  [P,P](\psi^1,\psi^2,\psi^3) = T(\psi^1,\psi^2,\psi^3)
  = \int T^{i_1 i_2\sigma_2i_3\sigma_3}
  \psi^1_{i_1}\partial_{\sigma_2}\psi^2_{i_2}\partial_{\sigma_3}\psi^3_{i_1}\, dx
\end{equation}
$T^{i_1i_2\sigma_2i_3\sigma_3}$ is defined up to total divergencies: this means
that three-vectors of the type
$\partial_x\left(T^{i_1\sigma_1i_2\sigma_2i_3\sigma_3}
  \partial_{\sigma_1}\psi^1_{i_1} \partial_{\sigma_2}\psi^2_{i_2}
  \partial_{\sigma_3}\psi^3_{i_1}\right)$ are zero. It immediately follows that
a local three-vector which is of order zero in one of its arguments is zero if
and only if its coefficients are zero.

The algorithm in Section~\ref{sec:algorithm} translates into the language of
differential operators as follows. Let us introduce the notation
\begin{equation}
  \label{eq:2}
  \tilde{\psi}^a_\alpha = \partial_x^{-1}(w_\alpha ^i\psi_i^a),
\end{equation}
where $a$ refers to the particular argument of the operator. Then, the vector
functions $\psi^1$, $\psi^2$, $\psi^3$ play the role of test vector functions
of the variables $x$, $y$, $z$ in the language of distributions.

\begin{enumerate}
\item The first step in Section~\ref{sec:algorithm} is not needed in the
  differential operator formalism, as it boils down to a change in the variable
  of integration (and its differential consequences).
\item The second step aims at bringing the nonlocal part of the three-vector in
  the reduced form~\eqref{eq:13}. To this aim, we remark that the reduced form
  of the distributions implies that there is no distribution of the type
  $\nu(x-y)$ acting on two vector test functions. This means effecting the
  following substitution (up to total divergencies)
  \begin{multline}
    \label{eq:1}
    e^\alpha w^{i}_\alpha\partial_x^{-1}
    \left(\pd{w^j_\alpha}{u^k_\tau}\partial_\tau (
    B^{kp\sigma}\partial_\sigma\psi^b_p
    + e^\alpha w^k_\alpha \tilde{\psi}^b_\alpha)^k\psi_j^c
    \right)\psi^a_i =
  \\
  - e^\alpha\tilde{\psi}^a_\alpha\left(\pd{w^j_\alpha}{u^k_\tau}\partial_\tau (
    B^{kp\sigma}\partial_\sigma\psi^b_p
    + e^\alpha w^k_\alpha \tilde{\psi}^b_\alpha)^k\psi_j^c
    \right)
  \end{multline}
  After such a substitution, we observe that the generic
  summands of~\eqref{eq:17} are of three types:
  \begin{gather}
    \label{eq:16}
    C^{\alpha\beta k}\tilde{\psi}^a_\alpha\tilde{\psi}^b_\beta \psi_k^c,
    \\ \label{eq:30}
    C^{\alpha kj\sigma}\tilde{\psi}^a_\alpha\partial_{\sigma}(\psi_j^b) \psi_k^c,
    \\ \label{eq:36}
    C^{kj\sigma i\tau}\partial_\tau(\psi^a_i)\partial_{\sigma}(\psi^b_j) \psi^c_k,
  \end{gather}
  where $C$'s are functions of $(u^i, u^i_\sigma)$. The reduced form of the
  three-vector in the formalism of differential operators amounts at bringing
  the operator to a canonical form where the arguments $\psi^a$, $\psi^b$,
  $\psi^c$ are a fixed sequence of integers (say, $1$, $2$, $3$) or its cyclic
  permutations (in the previous example, $3$, $1$, $2$ and $2$, $3$, $1$).
  This task can always be achieved by integration by parts that will produce
  the required terms plus extra terms.
\item The third step of the algorithm amounts at bringing the local part into a
  reduced form. This is achieved with the usual procedure of
  integrating by parts the three-vector with respect to one distinguished
  argument (say $\psi^3$) in such a way that the result will be of order zero
  in that argument.
\end{enumerate}

\section{Jacobi identity and Poisson Vertex Algebras}
\label{sec:nPVA}

\subsection{The Jacobi identity}
\label{sec:nPVA1}
\begin{definition}
  A (nonlocal) Poisson vertex algebra (PVA) is a differential algebra $(%
  \mathcal{A},\partial)$ endowed with a derivation $\partial$ and a bilinear
  operation $\{\cdot_\lambda\cdot\}\colon\mathcal{A}\otimes\mathcal{A}\to%
  \mathbb{R}((\lambda^{-1}))\otimes\mathcal{A}$ called a (nonlocal) \emph{$%
    \lambda$-bracket}, satisfying the following set of properties:
  \begin{enumerate}
  \item $\{\partial f_\lambda g\}=-\lambda\{f_\lambda g\}$ (left
    sesquilinearity),
  \item $\{f_\lambda \partial g\}=(\lambda+\partial)\{f_\lambda g\}$ (right
    sesquilinearity),
  \item $\{f_\lambda gh\}=\{f_\lambda g\}h +\{f_\lambda h\}g$ (left Leibnitz
    property),
  \item
    $\{fg_\lambda h\}=\{f_{\lambda+\partial} h\}g+\{g_{\lambda+\partial} h\} f$
    (right Leibnitz property),
  \item $\{g_\lambda f\}=-{}_\to\{f_{-\lambda-g}g\}$ (PVA skew-symmetry),
  \item
    $\{f_\lambda\{g_\mu h\}\}-\{g_\mu\{f_\lambda h\}\}=\{\{f_\lambda
    g\}_{\lambda+\mu} h\}$ (PVA-Jacobi identity).
  \end{enumerate}
\end{definition}

In the notation for the bracket, the symbol separating the two arguments is the
formal parameter of the expansion. We denote
\begin{equation*}
\{f_\lambda g\}=\sum_{s\leq S} C_s(f,g)\lambda^s,
\end{equation*}
with $C_s(f,g)\in\A$; the argument signals that each of the coefficients of the
expansion depends on the two elements $f$ and $g$ in $\A$. Such an expansion is
bounded by $0\leq s\leq S$ for \emph{local} PVAs and is not bounded from below
for \emph{nonlocal} PVAs.

The special notation used on the RHS of Property 4 is to be understood as
$$
\{f_{\lambda+\partial}g\}h=\sum_{s\leq S} C_s(f,g)(\lambda+\partial)^s h=\sum_{s,t}
\binom{s}{t}C_s(f,g) \partial^t h\,\lambda^{s-t}.
$$
Similarly, the RHS of Property 5 (the skewsymmetry) reads
$$
{}_\to\{f_{-\lambda-\partial}g\}=\sum_s(-\lambda-%
\partial)^sC_s(f,g).
$$

For a nonlocal $\lambda$ bracket, the three terms of
PVA-Jacobi identity do not necessarily belong to the same space, because of the
double infinite expansion of the brackets (in terms of $(\lambda,\mu)$, $%
(\mu,\lambda)$ and $(\lambda,\lambda+\mu)$, respectively). A bracket is said to
be \emph{admissible} if all the three terms can be (not uniquely) expanded as
\begin{equation*}
\{f_\lambda\{g_\mu h\}\}=\sum_{m\leq M}\sum_{n\leq N}\sum_{p\leq
	0}a_{m,n,p}\lambda^m\mu^n(\lambda+\mu)^p.
\end{equation*}
Only admissible brackets can define a nonlocal PVA. We denote the space
where the PVA-Jacobi identity of admissible brackets takes values by $%
V_{\lambda,\mu}$. This space can be decomposed by the total degree $d$ in $%
(\lambda,\mu,\lambda+\mu)$; elements of each homogeneous component $%
V^{(d)}_{\lambda,\mu}$ can be \emph{uniquely} expressed in the basis \cite%
{PV2}
\begin{align*}
\lambda^i\mu^{d-i}& & i&\in \mathbb{Z}, \\
\lambda^{d+i}(\lambda+\mu)^{-i}& & i&=\{1,2,\ldots\}.
\end{align*}
This filtration in the total degree $d$ and the subsequent choice of a basis
plays a crucial role in obtaining the normal form for the PVA-Jacobi identity.

The main result used to perform most of the computations is the so called
\emph{master formula}. Under the hypothesis that the differential algebra $\A$
is generated by the elements $(u^i)$, the $\lambda$-bracket between any two
elements of $\A$ is explicitly given by
\begin{equation} \label{eq:masterformula}
\{f_\lambda g\}= %\sum_{i,j=1}^\ell\sum_{l\geq 0}\sum_{m\geq 0}
\frac{\partial g}{%
	\partial u^j_\sigma}\left(\lambda+\partial\right)^\sigma\{u^i_{\lambda+%
	\partial}u^j\}\left(-\lambda-\partial\right)^\tau
\frac{\partial f}{\partial u^i_\tau}
\end{equation}
Thus, the structure of a PVA is defined by the matrix of the $\lambda$ brackets
between the generators $\{u^i_{\lambda}u^j\}=P^{ji}(\lambda)$. 

In the nonlocal case, expressions such as $(\lambda+\d)^p$ for $p<0$ arise in
$P^{ji}(\lambda+\d)$ from the master formula \eqref{eq:masterformula}. In such
cases, the rigorous approach -- working for any kind of nonlocality --
is to expand the negative powers of $(\lambda+\d)$ as
\begin{equation}
(\lambda+\d)^{-p}=\sum_{k\geq 0}\binom{-p}{k}\lambda^{-p-k}\d^k,\qquad\qquad p>0
\end{equation}
In the weakly nonlocal case this can be avoided, relying only on Properties
(1)-(4) of the lambda bracket. More details on this will be provided in Section
\ref{ssec:basis}.

\subsection{Dictionary: Poisson Vertex Algebras and differential
	operators}
\label{sec:PVA-voc}

The connection between the theory of PVA and Hamiltonian operator is given by
Theorem 4.8 in \cite[pag.~261]{PV2}. In short, there is a 1-1 correspondence
between $\lambda$-brackets of a (nonlocal) PVA and (pseudo)differential
Hamiltonian operators; the entries of the matrix $P^{ji}(\lambda)$ correspond
to the differential operator $P^{ij}$~\eqref{eq:10} after the formal
replacement of $\lambda$ by $\partial$.

More precisely, the equivalence between the expression of the Poisson bracket
\eqref{eq:7} and the expression of a $\lambda$-bracket according with the
master formula \eqref{eq:masterformula} is:
\begin{multline}\label{eq:5}
\{F,G\}_J =
\int \frac{\delta f}{\delta u^i}P^{ij\sigma}
\partial_\sigma\frac{\delta g}{\delta u^j} \, dx
\\
= \int\frac{\partial g}{\partial u^i_\sigma}\partial_\sigma
\left(P^{ij\tau}\partial_\tau
(-\partial)^\epsilon\frac{\partial f}{\partial u^j_\epsilon}
\right) dx =
\int\{f_\lambda g\}\big\vert_{\lambda=0}\,dx ,
\end{multline}
using \eqref{eq:masterformula}. The PVA-Jacobi identity for a triple of
generators $(u^i, u^j, u^k)$ can also be expressed by means of differential
operators. First of all, we compute the PVA-Jacobi identity using the master
formula; we have
\begin{gather}
\label{eq:PVA1}
\{u^i_\lambda\{u^j_\mu u^k\}\}=
\frac{\partial P^{kj}(\mu)}{\d u^l_\sigma}(\lambda+\d)^\sigma P^{li}(\lambda)
\\
\label{eq:PVA2}
\{u^j_\mu\{u^i_\lambda u^k\}\}=
\frac{\partial P^{ki}(\lambda)}{\d u^l_\sigma}(\mu+\d)^\sigma P^{lj}(\mu)
\\
\label{eq:PVA3}
\{\{u^i_\lambda u^j\}_{\lambda+\mu}u^k\}
= P^{kl}(\lambda+\mu+\d)(-\lambda-\mu-\d)^\sigma
\frac{\partial P^{ji}(\lambda)}{\d u^l_\sigma}
\end{gather}
The PVA-Jacobi identity is $J^{ijk}_{\lambda,\mu}(P,P)=$ \eqref{eq:PVA1} $-$
\eqref{eq:PVA2} $-$ \eqref{eq:PVA3} $=0$. We evaluate the expression on three
covectors $\psi^1_i\psi^2_j\psi^3_k$, and regard each power of $\lambda$ as
derivations acting on $\psi^1$, and each power of $\mu$ as derivations acting
on $\psi^2$.  Then, the three summands correspond to
\begin{gather}
\langle \psi^3,\ell_{P,\psi^1}(P\psi^2)\rangle,
\\
\langle\psi^3,\ell_{P,\psi^2}(P\psi^1)\rangle,
\\
\langle\psi^3,P\,\ell^*_{P,\psi^1}(\psi^2)\rangle,
\end{gather}
respectively, and the PVA-Jacobi identity is the vanishing of the Schouten
bracket $[P,P]$ in the form of~\eqref{eq:18}.

\subsection{The algorithm}
\label{sec:algorithm-1}

For the local case, the expression of the PVA-Jacobi identity is a polynomial
in $\lambda$ and $\mu$, and the vanishing of the coefficients of
$\lambda^p\mu^q$ corresponds to the vanishing of the
coefficients for $\partial^p(\psi^1_i)\partial^q(\psi^2_j)\psi^3_k$.

In the nonlocal case, the PVA-Jacobi identity is a Laurent series in
$\lambda^{-1}$, $\mu^{-1}$ and $(\lambda+\mu)^{-1}$ living in the space
$V_{\lambda,\mu}$ defined in Section \ref{sec:nPVA1}: in the weakly nonlocal
case, these coefficients come, respectively, from the expansion of
$(\lambda+\d)^{-1}$, $(\mu+\d)^{-1}$, $(\lambda+\mu+\d)^{-1}$.

From the computation of the PVA-Jacobi identity we obtain seven types of terms
including one or two nonlocal factors, together with the pure local terms; each
of them corresponds to the types of summands in the three-vector of the
Schouten identity in \eqref{eq:19}, as detailed in \eqref{eq:16} and
following. They are
\begin{enumerate}
\item $A^{ijk}\lambda^p\mu^q$ with $p,q\geq 0$, corresponding to
  $\d^p(\psi_i^1)\d^q(\psi_j^2)\psi^3_k$;
\item $w^k(\lambda+\mu+\d)^{-1}A^{ij}\lambda^p$ with $p\geq0$, corresponding to
  $\d^p(\psi_i^1)\psi_j^2\tilde{\psi}^3$;
\item $w^k(\lambda+\mu+\d)^{-1}A^{ij}\mu^p$ with $p>0$, corresponding to
  $\psi_i^1\d^p(\psi_j^2)\tilde{\psi}^3$;
\item $[(\lambda+\d)^{-1}w^i]A^{jk}\mu^p$ with $p\geq0$, corresponding to
  $\tilde{\psi}^1\d^p(\psi_j^2)\psi^3_k$;
\item $[(\mu+\d)^{-1}w^j]A^{ki}\lambda^p$ with $p\geq 0$, corresponding to
  $\d^p(\psi_i^1)\tilde{\psi}^2\psi^3_k$;
\item $w^k(\lambda+\mu+\d)^{-1}A^j(\lambda+\d)^{-1}w^i$, corresponding to
  $\tilde{\psi}^1\psi_j^2\tilde{\psi}^3$;
\item $w^k(\lambda+\mu+\d)^{-1}A^i(\mu+\d)^{-1}w^j$, corresponding to
  $\psi_i^1\tilde{\psi}^2\tilde{\psi}^3$;
\item $[(\lambda+\d)^{-1}w^i]A^k[(\mu+\d)^{-1}w^j]$, corresponding to
  $\tilde{\psi}^1\tilde{\psi}^2\psi^3_k$.
\end{enumerate}
The square brackets denote that the differential operators obtained by the
expansion of the pseudodifferential operator do not act outside them.

Note that the expansion of the terms 3 and 7 is not expressed in the basis for
$V_{\lambda,\mu}$ we have chosen; on the other hand, terms 3 and 5 do not
correspond to the choice of coefficients for the normalization algorithm of the
previous Sections (when one takes the cyclic ordering
$\tilde{\psi}^a\d^p(\psi^b)\psi^c$).

We give a different treatment of the terms including at most one nonlocal
expression and of the ones with two: in the first case, we bring them to a form
whose expansion is automatically expressed in our chosen basis for
$V_{\lambda,\mu}$; in the second case, we show that the vanishing of the term
7, together with the other ones, is equivalent to the vanishing of the
corresponding terms in the expansion on the basis.

Finally, we comment on the equivalence between the vanishing of the PVA-Jacobi
identity on our chosen basis and as a result of the normalization algorithm of
the previous Sections.

\begin{proposition}
  The terms of type $w^k(\lambda+\mu+\d)^{-1}A^{ij}\mu^p$ can be brought to a
  combination of terms of type $A^{ijk}\lambda^p\mu^q$ with $p,q\geq 0$ and
  $w^k(\lambda+\mu+\d)^{-1}A^{ij}\lambda^p$, reducing the PVA-Jacobi identity
  to the expansion of seven terms.
\end{proposition}
\begin{proof}
  From the expansion
  $(\lambda+\mu+\d)^p=\sum_{l=0}^p\binom{p}{l}\mu^{p-l}(\lambda+\d)^l$ we can
  rewrite a term of the form $w^k(\lambda+\mu+\d)^{-1}A^{ij}\mu^p$ as
	$$
	w^k(\lambda+\mu+\d)^{-1}\left[(\lambda+\mu+\d)^p A^{ij}
          -\sum_{l=0}^{p-1}\binom{p}{l}\mu^{l}(\lambda+\d)^{p-l}A^{ij}\right],
	$$
	which gives
	$$
	w^k(\lambda+\mu+\d)^{p-1}
        A^{ij}-w^k(\lambda+\mu+\d)^{-1}\left[\sum_{l=0}^{p-1}\binom{p}{l}\mu^{l}(\lambda+\d)^{p-l}A^{ij}\right].
	$$
	The expression in the square bracket has top degree $p-1$ in
        $\mu$. Repeating the operation we obtain only local terms or terms of
        the type 2.
\end{proof}
\begin{theorem}
  The PVA-Jacobi identity, expressed using terms of the type $1$, $2$, $4$ --
  $8$ as above, can always be expressed in the space
  $V^{(d)}_{\lambda,\mu}$. This latter expression vanishes if and only if the
  former does.
\end{theorem}

\begin{proof}
  Expressing the PVA-Jacobi identity in the space $V^{(d)}_{\lambda\,\mu}$, for
  all for $d\leq D$, means expanding it on the basis $\lambda^p \mu^{d-p}$,
  $p\in \mathbb{Z}$ and $(\lambda+\mu)^{-p}\lambda^{d+p}$, $p>0$. Terms of type
  1 do not need to be expanded, as they are already expressed in the basis for
  $V^{(d)}_{\lambda,\mu}$, $d\geq 0$.
	
  For the types with one nonlocal term, namely $2$, $4$ and $5$ in the previous
  list, the expansions of the pseudodifferential operators give the series
  \begin{gather}\label{eq:lm_non_l}
    \sum_{t\geq
      0}(-1)^tw^k(\lambda+\mu)^{-t-1}\lambda^p\d^tA^{ij}\\\label{eq:lm_non_m}
    \sum_{t\geq 0}(-1)^t A^{jk}\mu^p\lambda^{-t-1}\d^tw^i\\\label{eq:l_non_m}
    \sum_{t\geq 0}(-1)^t A^{ki}\lambda^p\mu^{-t-1}\d^tw^j
  \end{gather}
  which are in our chosen basis of $V^{(d)}_{\lambda,\mu}$, for $d\leq
  p-1$. The vanishing of the $t=0$ term in the expansion is a sufficient and
  necessary condition for the vanishing of the whole series: all the subsequent
  terms in the expansion vanish if the first one does, and it must vanish
  because that is the only one in $V^{(p-1)}_{\lambda,\mu}$ containing the
  factor $(\lambda+\mu)^{-1}$ (resp. $\lambda^{-1}$ and $\mu^{-1}$). It is
  hence enough to check (or impose) the vanishing of the coefficients $A$ or
  $w$. However, the vanishing
  of the $w$ terms coincides with the dropping of the nonlocal part of the
  $\lambda$ bracket, so the condition is only on $A$'s.
	
  A similar point can be made for the types 6 and 8 with the double
  nonlocality: their expansion is expressed in our chosen basis and starts,
  respectively, with $(\lambda+\mu)^{-1}\lambda^{-1}$ and
  $\lambda^{-1}\mu^{-1}$ in $V^{(-2)}_{\lambda,\mu}$. The expansion of the term
  7 starts with $A^iw^kw^j(\lambda+\mu)^{-1}\mu^{-1}$, which is not an element
  in the basis of $V^{(-2)}_{\lambda,\mu}$. However, this is a term we can
  rearrange as an infinite series
  \begin{displaymath}
    A^iw^kw^j\left(\lambda^{-1}\mu^{-1}-(\lambda+\mu)^{-2}
      + \sum_{m>0}c_m(\lambda+\mu)^{-2-m}\lambda^m\right)
  \end{displaymath}
  for some fixed constants $c_m$.
	
  Note that elements in $V^{(-2)}_{\lambda,\mu}$ could be obtained by the
  expansions (for $t=1$) of the previous terms with only one
  nonlocality. However, the vanishing of the elements in
  $V^{(-1)}_{\lambda,\mu}$ implies their vanishing, too, and hence we can focus
  on the terms arising from the expansion of double nonlocalities only.
	
  It is straightforward to see that we get only one expression in front of
  $(\lambda+\mu)^{-1}\lambda^{-1}$ (from type 6) and $(\lambda+\mu)^{-2}$ (from
  our rearrangement of type 8); on the other hand, there could be two sources
  of terms of the form $\lambda^{-1}\mu^{-1}$.  The vanishing of either $A$ or $w$ for all $i,j,k$ in the first two cases is a necessary and
  sufficient condition; once that this has been imposed or checked, the only
  surviving class of terms of the form $\lambda^{-1}\mu^{-1}$ comes from the
  expansion of 7.
	
  Since the vanishing of $w$ is equivalent to
  simply dropping the nonlocal term of the operator, the condition we
  need to consider is only the vanishing of the expressions $A$'s.
\end{proof}

\begin{remark}
  The above theorem has two important consequences.
  \begin{enumerate}
  \item This algorithm always yields a divergence-free form of the Jacobi
    identity; this means that the Jacobi identity holds if and only if the
    coefficient of the Laurent series in the spaces $V^{(d)}_{\lambda,\mu}$
    vanish.
  \item There is no need to expand in Laurent series: indeed, the expansion is
    always ruled by the zeroth-order coefficients, which are just the
    coefficients of the terms $2$, $4$ -- $8$.
  \end{enumerate}
\end{remark}
\begin{remark}
  Writing the PVA-Jacobi identity on our chosen basis of
  $V^{(d)}_{\lambda,\mu}$ for $d\geq -2$ yields a different result than the one
  obtained with the algorithm described in Section \ref{sec:algorithm-2}. For
  the terms with one nonlocality, indeed, the PVA-Jacobi identity produces the
  coefficients corresponding to $\d^p(\psi_i^1)\psi_j^2\tilde{\psi}^3$,
  $\d^p(\psi_j^2)\psi_k^3\tilde{\psi}^1$ and
  $\d^p(\psi_i^1)\psi_j^3\tilde{\psi}^2$, while the latter is replaced by
  $\d^p(\psi_i^3)\psi_j^1\tilde{\psi}^2$ in Section \ref{sec:algorithm-2}.
	
  Nevertheless, the sets of condition given by the vanishing of the
  coefficients in front of the terms obtained with the two different algorithms
  are equivalent. Let us demonstrate it assuming that the terms of type 5 are
  \begin{equation}\label{eq:NormNonLoc1}
    A_2\,\lambda^2(\mu+\d)^{-1}w+A_1\,\lambda(\mu+\d)^{-1}w+A_0\,(\mu+\d)^{-1}w,
  \end{equation}
  corresponding to
  \begin{equation}
    A_2\,\tilde{\psi}^2\d^2(\psi^1)\psi^3+A_1\tilde{\psi}^2\d(\psi^1)\psi^3
    +A_0\,\tilde{\psi}^2\psi^1\psi^3.
  \end{equation}
  This latest expression is equivalent, up to total derivatives, to
  \begin{multline}\label{eq:NormNonLoc2}
    A_2\tilde{\psi}^2\psi^1\d^2(\psi^3)+\left(2\d
      A_2-A_1\right)\tilde{\psi}^2\psi^1\d(\psi^3)+\left(A_0+\d^2A_2-\d
      A_1\right)\tilde{\psi}^2\psi^1\psi^3\\+\text{local terms}.
  \end{multline}
  The vanishing of expression \eqref{eq:NormNonLoc2} at top degree implies the
  vanishing of the lower degree coefficients, being hence equivalent to the
  vanishing of \eqref{eq:NormNonLoc1}.

  The same result can be obtained in the framework of Poisson vertex algebras
  introducing the symbol $\nu=-\lambda-\mu-\d$, representing derivations acting
  on $\psi^3$ \cite[Section~4.1]{dSK13}.
\end{remark}

\section{Weakly nonlocal PBHT and distributions}

\subsection{Calculation of the Jacobi identity}
\label{sec:calc-jacobi-ident}

In this section we will consider, as an example, weakly nonlocal Poisson
bracket of hydrodynamic type, of the form~\eqref{eq:8}. In the language of
distribution it has the form
\begin{equation}
P^{ij}_{x,y}=g^{ij}({\bf u}(x))\delta'_{xy}+ \Gamma^{ij}_k({\bf
    u}(x))u^k_x\delta_{xy}+ w^i_s({\bf u}(x))u^s_x\nu_{xy}w^j_t({\bf
    u}(y))u^t_y\label{eq:20}
\end{equation}
(we will use only one `tail summand' to make calculations simpler)
where $\delta_{xy}=\delta(x-y)$ e $\nu_{xy}=\nu(x-y)$.  We assume $g$ to be non
degenerate. In what follows, an index after a comma denotes a partial
derivative with respect to the corresponding field variable, \emph{e.g.}
$g^{ij}_{,k}= \partial g^{ij}/\partial u^k$.

From the skew-symmetry the two conditions \eqref{eq:52}, \eqref{eq:56} follow,
namely: $g^{ij}=g^{ji}$ and $g^{ij}_{,k}=\Gamma^{ij}_k+\Gamma^{ji}_k$.
We apply now the reducing procedure explained in
Section~\ref{sec:algorithm}. Since $P^{ij}_{xy}$ depend only on $u(x)$ and
$u_x$ each sum in \eqref{SBDZ} contains only two terms.  The Jacobi identity
can be rewritten as
\begin{multline}
\frac{\partial P^{ij}_{x,y}}{\partial u^{l}(x)}P^{lk}_{x,z}
+ \frac{\partial P^{ij}_{x,y}}{\partial u^{l}(y)} P^{lk}_{y,z}
+\frac{\partial P^{ki}_{z,x}}{\partial u^{l}(z)} P^{lj}_{z,y}
+ \frac{\partial P^{ki}_{z,x}}{\partial u^{l}(x)} P^{lj}_{x,y}+
\\
+\frac{\partial P^{jk}_{y,z}}{\partial u^{l}(y)} P^{li}_{y,x}
+ \frac{\partial P^{jk}_{y,z}}{\partial u^{l}(z)} P^{li}_{z,x}
+\frac{\partial P^{ij}_{x,y}}{\partial u^{l}_x} \partial_x P^{lk}_{x,z}
+ \frac{\partial P^{ij}_{x,y}}{\partial u^{l}_y} \partial_y P^{lk}_{y,z}+
\\
\label{SBDZ2}
\frac{\partial P^{ki}_{z,x}}{\partial u^{l}_z} \partial_z P^{lj}_{z,y}
+ \frac{\partial P^{ki}_{z,x}}{\partial u^{l}_x} \partial_x P^{lj}_{x,y}
+\frac{\partial P^{jk}_{y,z}}{\partial u^{l}_y} \partial_y P^{li}_{y,x}
+ \frac{\partial P^{jk}_{y,z}}{\partial u^{l}_z} \partial_z P^{li}_{z,x}=0
\end{multline}

\subsection{Calculation of the reduced form}
\label{sec:calc-reduc-form}

The first summand in \eqref{SBDZ2} is
\begin{multline}\label{eq:26}
  \frac{\partial P^{ij}_{x,y}}{\partial u^{l}(x)}P^{lk}_{x,z}=
  \left(g^{ij}_{,l}\delta'_{xy}+
    \Gamma^{ij}_{s,l}u^s_x\delta_{xy}+
    w^i_{s,l}u^s_x\nu_{xy}w^j_t\,u^t_y\right)\cdot
  \\
  \cdot \left(g^{lk}(x)\delta'_{xz}+\Gamma^{lk}_tu^t_x\delta_{xz}
    + w^l_su^s_x\nu_{xz}w^k_tu^t_z\right)
\end{multline}

The coefficients of the reduced form are listed below.
\begin{itemize}
\item The coefficient of $\delta'_{xy}\delta'_{xz}$ is
  $g^{lk}g^{ij}_{,l}$.
\item The coefficient of $\nu_{xy}\nu_{xz}$ is
  $w^i_{s,l} w^{l}_{m}u^{m}_xu^s_xw^j_tu^t_yw^{k}_{n}u^{n}_z$.
\item The coefficient of $\delta'_{xy}\delta_{xz}$ is
  $g^{ij}_{,l}\Gamma^{lk}_tu^t_x$.
\item The coefficient of $\delta_{xy}\delta'_{xz}$ is
$g^{lk}\Gamma^{ij}_{s,l}u^s_x$.
\item The coefficient of $\delta_{xy}\delta_{xz}$ is
$\Gamma^{ij}_{s,l}\Gamma^{lk}_tu^t_xu^s_x - g^{ij}_{,l}w^l_su^s_xw^k_tu^t_x$.
\item The coefficient of $\delta_{yx}\nu_{yz}$ is
  $-\d_y(g^{ij}_{,l}w^l_su^s_y)w^k_tu^t_z +
  \Gamma^{ij}_{s,l}w^l_ru^r_xu^s_xw^k_tu^t_z$.
\item The coefficient of $\delta_{xz}\nu_{xy}$ is
$w^i_{s,l}\Gamma^{lk}_ru^r_xu^s_xw^j_tu^t_y$.
\item The coefficient of $\delta'_{yx}\nu_{yz}$ is
$ - g^{ij}_{,l}w^l_su^s_yw^k_tu^t_z$.
\item The coefficient of $\delta'_{xz}\nu_{xy}$ is
$g^{lk}w^i_{s,l}u^s_xw^j_tu^t_y$.
\end{itemize}

The second summand in \eqref{SBDZ2} is
\begin{multline}\label{eq:27}
\frac{\partial P^{ij}_{x,y}}{\partial u^{l}(y)}P^{lk}_{y,z}=
w^i_s u^s_x\nu_{xy}g^{lk} w^j_{t,l}u^t_y\delta'_{yz}+
\\
+ w^i_s u^s_x\nu_{xy}w^j_{t,l}\Gamma^{lk}_mu^m_yu^t_y\delta_{yz}
+ w^i_s u^s_x\nu_{xy} w^j_{t,l}w^l_mu^m_yu^t_y\nu_{yz}w^k_nu^n_z.
\end{multline}

The coefficients of the reduced form are listed below:
\begin{itemize}
\item The coefficient of $\delta_{xy}\delta_{xz}$ is
${w^i_s u^s_xg^{lk} w^j_{t,l}u^t_x}$.
\item  The coefficient of $\nu_{xz}\delta'_{zy}$ is
$- w^i_s u^s_xg^{lk} w^j_{t,l}u^t_z$.
\item The coefficient of $\nu_{xz}\delta_{zy}$ is
${- w^i_s u^s_x\d_z\left(g^{lk} w^j_{t,l}u^t_z\right)
+ w^i_s u^s_x w^j_{t,l}\Gamma^{lk}_mu^m_yu^t_y}$.
\item The coefficient of $\nu_{xy}\nu_{yz}$ is
${ w^i_s u^s_x w^j_{t,l}w^l_mu^m_yu^t_yw^k_nu^n_z}$.
\end{itemize}

The third summand in \eqref{SBDZ2} is
\begin{multline}
\frac{\partial P^{ki}_{z,x}}{\partial u^{l}(z)} P^{lj}_{z,y}=
  g^{lj} g^{ki}_{,l}\delta'_{zx}\delta'_{zy}
  + g^{ki}_{,l}\Gamma^{lj}_tu^t_z\delta'_{zx}\delta_{zy}
  + g^{lj}\Gamma^{ki}_{s,l}u^s_z\delta_{zx}\delta'_{zy}
\\
  + \Gamma^{ki}_{s,l}\Gamma^{lj}_tu^t_zu^s_z\delta_{zx}\delta_{zy}
  + g^{ki}_{,l}w^l_su^s_z\delta'_{zx}\nu_{zy}w^j_tu^t_y
  + \Gamma^{ki}_{s,l}w^l_ru^r_zu^s_z\delta_{zx}\nu_{zy}w^j_tu^t_y
\\
  + g^{lj} w^k_{s,l}u^s_z\delta'_{zy}\nu_{zx}w^i_tu^t_x
  + w^k_{s,l}\Gamma^{lj}_tu^t_zu^s_z\delta_{zy}\nu_{zx}w^i_tu^t_x
\\
  + w^k_{s,l}w^{l}_{m}u^{m}_zu^s_z\nu_{zx}w^i_tu^t_x\nu_{zy}w^{j}_{n}u^{n}_y.
\end{multline}

The coefficients of the reduced form are listed below.

\begin{itemize}
\item The coefficient of $\delta''_{xy}\delta_{xz}$ is
$-g^{lj}g^{ki}_{,l}$.
\item The coefficient of $\delta'_{xy}\delta'_{xz}$ is
$-g^{lj} g^{ki}_{,l}$.
\item The coefficient of $\delta'_{xy}\delta_{xz}$ is
  $-\d_x\left(g^{lj}g^{ki}_{,l}\right) - g^{ki}_{,l}\Gamma^{lj}_tu^t_x
  + g^{lj}\Gamma^{ki}_{s,l}u^s_x$.
\item The coefficient of $\delta_{xy}\delta'_{xz}$ is
$ - g^{ki}_{,l}\Gamma^{lj}_tu^t_x$.
\item The coefficient of $\delta_{xy}\delta_{xz}$ is
  $ - \d_x\left(g^{ki}_{,l}\Gamma^{lj}_tu^t_x\right)
  + \Gamma^{ki}_{s,l}\Gamma^{lj}_tu^t_xu^s_x
  - g^{ki}_{,l}w^l_su^s_xw^j_tu^t_x$.
\item The coefficient of $\delta_{xz}\nu_{xy}$ is
$ - \d_x\left(g^{ki}_{,l}w^l_su^s_x\right)w^j_tu^t_y
+  \Gamma^{ki}_{s,l}w^l_ru^r_xu^s_xw^j_tu^t_y$.
\item The coefficient of  $\delta_{zy}\nu_{zx}$ is
$w^k_{s,l}\Gamma^{lj}_tu^t_zu^s_zw^i_tu^t_x$.
\item The coefficient of $\delta'_{zy}\nu_{zx}$ is
$g^{lj} w^k_{s,l}u^s_zw^i_tu^t_x$.
\item The coefficient of $\delta'_{xz}\nu_{xy}$ is
$ - g^{ki}_{,l}w^l_su^s_xw^j_tu^t_y$.
\item The coefficient of $\nu_{zx}\nu_{zy}$ is
$w^k_{s,l}w^{l}_{m}u^{m}_zu^s_zw^i_tu^t_xw^{j}_{n}u^{n}_y$.
\end{itemize}

The fourth summand in \eqref{SBDZ2} is
\begin{multline}
  \frac{\partial P^{ki}_{z,x}}{\partial u^{l}(x)} P^{lj}_{x,y}=
  w^k_s u^s_z\nu_{zx}g^{lj} w^i_{t,l}u^t_x\delta'_{xy}
  \\
  + w^k_s u^s_z\nu_{zx} w^i_{t,l}\Gamma^{lj}_mu^m_xu^t_x\delta_{xy}
  + w^k_s u^s_z\nu_{zx} w^i_{t,l}w^l_mu^m_xu^t_x\nu_{xy}w^j_nu^n_y.
\end{multline}

The coefficients of the reduced form are listed below:
\begin{itemize}
\item The coefficient of $\delta_{xy}\delta_{xz}$ is
$w^k_s u^s_xg^{lj} w^i_{t,l}u^t_x$.
\item The coefficient of $\nu_{zy}\delta'_{yx}$ is
$ - w^k_s u^s_zg^{lj} w^i_{t,l}u^t_y$.
\item The coefficient of $\nu_{zy}\delta_{xy}$ is
$ - w^k_s u^s_z\d_y\left(g^{lj} w^i_{t,l}u^t_y\right) - w^k_s
u^s_z w^i_{t,l}\Gamma^{lj}_mu^m_xu^t_x$.
\item The coefficient of $\nu_{zx}\nu_{xy}$ is
$w^k_s u^s_z w^i_{t,l}w^l_mu^m_xu^t_x w^j_nu^n_y$.
\end{itemize}

The fifth summand in \eqref{SBDZ2} is
\begin{multline}
  \frac{\partial P^{jk}_{y,z}}{\partial u^{l}(y)} P^{li}_{y,x}=
  g^{li} g^{jk}_{,l}\delta'_{yz}\delta'_{yx}
  + g^{jk}_{,l}\Gamma^{li}_tu^t_y\delta'_{yz}\delta_{yx}
  + g^{li}\Gamma^{jk}_{s,l}u^s_y\delta_{yz}\delta'_{yx}
  \\
  + \Gamma^{jk}_{s,l}\Gamma^{li}_tu^t_yu^s_y\delta_{yz}\delta_{yx}
  + g^{jk}_{,l}w^l_su^s_y\delta'_{yz}\nu_{yx}w^i_tu^t_x
  + \Gamma^{jk}_{s,l}w^l_su^s_yu^s_y\delta_{yz}\nu_{yx}w^i_tu^t_x
  \\
  + g^{li} w^j_{s,l}u^s_y\delta'_{yx}\nu_{yz}w^k_tu^t_z+
  + w^j_{s,l}\Gamma^{li}_tu^t_yu^s_y\delta_{yx}\nu_{yz}w^k_tu^t_z
  \\
  + w^j_{s,l}w^{l}_{m}u^{m}_yu^s_y\nu_{yz}w^k_tu^t_z\nu_{yx}w^{i}_{n}u^{n}_x.
\end{multline}

The coefficients of the reduced form are
\begin{itemize}
\item The coefficient of $\delta'_{xy}\delta'_{xz}$ is
$ - g^{li} g^{jk}_{,l}$.
\item The coefficient of $\delta_{xy}\delta''_{xz}$ is
$ - g^{li} g^{jk}_{,l}$.
\item The coefficient of $\delta_{xy}\delta'_{xz}$ is
  $ - \d_x\left(g^{li} g^{jk}_{,l}\right)
  + g^{jk}_{,l}\Gamma^{li}_tu^t_y
  - g^{li} \Gamma^{jk}_{s,l}u^s_x$.
\item The coefficient of $\delta'_{xy}\delta_{xz}$ is
$-g^{li} \Gamma^{jk}_{s,l}u^s_x$.
\item The coefficient of $\delta_{xy}\delta_{xz}$ is
  $-\d_x\left(g^{li} \Gamma^{jk}_{s,l}u^s_x\right)
  + \Gamma^{jk}_{s,l}\Gamma^{li}_tu^t_xu^s_x
  - g^{jk}_{,l}w^l_su^s_xw^i_tu^t_x$.
\item The coefficient of $\delta_{zy}\nu_{zx}$ is
  $ - \d_z\left( g^{jk}_{,l}w^l_su^s_z\right)w^i_tu^t_x
  + \Gamma^{jk}_{s,l}w^l_su^s_yu^s_yw^i_tu^t_x$.
\item The coefficient of $\delta'_{zy}\nu_{zx}$ is
$ - g^{jk}_{,l}w^l_su^s_zw^i_tu^t_x$.
\item The coefficient of $\delta'_{yx}\nu_{yz}$ is
$g^{li} w^j_{s,l}u^s_yw^k_tu^t_z$.
\item The coefficient of $\nu_{yz}\nu_{yx}$ is
$ w^j_{s,l}w^{l}_{m}u^{m}_yu^s_yw^k_tu^t_zw^{i}_{n}u^{n}_x$.
\end{itemize}

The sixth summand in \eqref{SBDZ2} is
\begin{multline}
  \frac{\partial P^{jk}_{y,z}}{\partial u^{l}(z)} P^{li}_{z,x}=
  w^j_s u^s_y\nu_{yz}g^{li} w^k_{t,l}u^t_z\delta'_{zx}
  \\
  + w^j_s u^s_y\nu_{yz} w^k_{t,l}\Gamma^{li}_mu^m_zu^t_z\delta_{zx}
  + w^j_s u^s_y\nu_{yz} w^k_{t,l}w^l_mu^m_zu^t_z\nu_{zx}w^i_nu^n_x.
\end{multline}

The coefficients of the reduced form are listed below
\begin{itemize}
\item The coefficient of $\delta_{xy}\delta_{xz}$ is
$w^j_s u^s_xg^{li} w^k_{t,l}u^t_x$.
\item The coefficient of $\nu_{xy}\delta'_{xz}$ is
$w^j_s u^s_yg^{li} w^k_{t,l}u^t_x$.
\item The coefficient of $\nu_{xy}\delta_{xz}$ is
  $ + w^j_s u^s_y\d_x\left(g^{li} w^k_{t,l}u^t_x\right)
  - w^j_s u^s_y w^k_{t,l}\Gamma^{li}_mu^m_xu^t_x$.
\item The coefficient of $\nu_{yz}\nu_{zx}$ is
$+w^j_s u^s_y w^k_{t,l}w^l_mu^m_zu^t_zw^i_nu^n_x$.
\end{itemize}

The seventh summand in \eqref{SBDZ2} is
\begin{multline}
\frac{\partial P^{ij}_{x,y}}{\partial u^{l}_x}\d_x P^{lk}_{x,z}=\\
\left(\Gamma^{ij}_l\delta_{xy}+ w^i_l\nu_{xy}w^j_t\,u^t_y\right)
\d_x\left(g^{lk}(x)\delta'_{xz}+\Gamma^{lk}_tu^t_x\delta_{xz}
  + w^l_su^s_x\nu_{xz}w^k_tu^t_z\right).
\end{multline}

The coefficients of the reduced form are listed below:
\begin{itemize}
\item The coefficient of $\delta_{xy}\delta''_{xz}$ is
$g^{lk}\Gamma^{ij}_l$.
\item The coefficient of $\delta''_{xz}\nu_{xy}$ is
$w^i_lg^{lk}w^j_r\,u^r_y$.
\item The coefficient of $\delta_{xy}\delta'_{xz}$ is
$\Gamma^{ij}_l\Gamma^{lk}_tu^t_x+\Gamma^{ij}_l g^{lk}_{,s}u^s_x$.
\item The coefficient of $\delta_{xy}\delta_{xz}$ is
  $\Gamma^{ij}_l \Gamma^{lk}_{t,s}u^s_xu^t_x
  +\Gamma^{ij}_l\Gamma^{lk}_tu^t_{xx}
  +\Gamma^{ij}_lw^l_su^s_xw^k_tu^t_x$.
\item The coefficient of $\delta_{yx}\nu_{yz}$ is
  $\Gamma^{ij}_l w^l_{s,m}u^m_xu^s_xw^k_tu^t_z
  + \Gamma^{ij}_lw^l_su^s_{xx} w^k_tu^t_z$.
\item The coefficient of $\delta_{xz}\nu_{xy}$ is
  \begin{displaymath}
    w^i_l\Gamma^{lk}_{t,s}u^s_xu^t_xw^j_r\,u^r_y
  + w^i_l\Gamma^{lk}_tu^t_{xx}w^j_ru^r_y
  + w^i_l w^l_su^s_xw^k_tu^t_xw^j_r\,u^r_y.
\end{displaymath}
\item The coefficient of $\delta'_{xz}\nu_{xy}$ is
$w^i_l g^{lk}_{,s}u^s_xw^j_r\,u^r_y+w^i_l\Gamma^{lk}_tu^t_xw^j_r\,u^r_y$.
\item The coefficient of $\nu_{xz}\nu_{xy}$ is
  $w^i_l w^l_{s,m}u^m_xu^s_xw^k_tu^t_zw^j_r\,u^r_y
  + w^i_l w^l_su^s_{xx}w^k_tu^t_zw^j_r\,u^r_y$.
\end{itemize}

The eighth summand in \eqref{SBDZ2} is
\begin{eqnarray*}
&&\frac{\partial P^{ij}_{x,y}}{\partial
   u^{l}_y}\partial_yP^{lk}_{y,z}=w^i_su^s_x\nu_{xy}w^j_l\d_y\left(
   g^{lk}\delta'_{yz}+\Gamma^{lk}_tu^t_y\delta_{yz}+
   w^l_su^s_y\nu_{yz}w^k_tu^t_z\right).
\end{eqnarray*}

The coefficients of the reduced form are listed below:
\begin{itemize}
\item The coefficient of $\delta_{xy}\delta_{xz}$ is
  $w^i_su^s_xw^j_l g^{lk}_{,m}u^m_x
  - w^i_su^s_x\d_x(w^j_lg^{lk})
  + w^i_su^s_xw^j_l\Gamma^{lk}_tu^t_x$.
\item The coefficient of $\nu_{xz}\delta'_{zy}$ is
  \begin{displaymath}
    - w^i_su^s_xw^j_l g^{lk}_{,m}u^m_z
  + 2 w^i_su^s_x\d_z(w^j_lg^{lk})
  - w^i_su^s_xw^j_l\Gamma^{lk}_tu^t_z.
\end{displaymath}
\item The coefficient of $\nu_{xz}\delta''_{zy}$ is
$w^i_su^s_xw^j_lg^{lk}$.
\item The coefficient of $\nu_{xz}\delta_{zy}$ is
  \begin{multline*}
    w^i_su^s_x\big(\d_z^2(w^j_lg^{lk})-\d_z\left(w^j_l g^{lk}_{,m}u^m_z\right)
      +w^j_l \Gamma^{lk}_{t,m}u^m_zu^t_z
      \\
      + w^j_l\Gamma^{lk}_tu^t_{zz}\delta_{zy}-\d_z(w^j_l\Gamma^{lk}_tu^t_z)
      + w^j_lw^l_ru^r_zw^k_tu^t_z\big).
  \end{multline*}
\item The coefficient of  $\nu_{xy}\nu_{yz}$ is
  $w^i_su^s_xw^j_l w^l_{r,m}u^m_yu^r_yw^k_tu^t_z
  + w^i_su^s_xw^j_lw^l_ru^r_{yy}w^k_tu^t_z$.
\item The coefficient of $\delta_{xy}\delta'_{xz}$ is
$w^i_su^s_xw^j_lg^{lk}$.
\item The coefficient of $\delta'_{xy}\delta_{xz}$ is
$ - w^i_su^s_xw^j_lg^{lk}$.
\end{itemize}

The ninth summand in \eqref{SBDZ2} is
\begin{equation}
\begin{split}
\frac{\partial P^{ki}_{z,x}}{\partial u^{l}_z} & \partial_z P^{lj}_{z,y}=
\Gamma^{ki}_l\f{\d g^{lj}}{\d u^s}u^s_z\delta_{zx}\delta'_{zy}
+g^{lj}\Gamma^{ki}_l\delta_{zx}\delta''_{zy}
\\
&+\Gamma^{ki}_l\Gamma^{lj}_{t,s}u^s_zu^t_z\delta_{zx}\delta_{zy}
+\Gamma^{ki}_l\Gamma^{lj}_tu^t_{zz}\delta_{zx}\delta_{zy}
+\Gamma^{ki}_l\Gamma^{lj}_tu^t_z\delta_{zx}\delta'_{zy}
\\
&+\Gamma^{ki}_l w^l_{s,m}u^m_zu^s_z\delta_{zx}\nu_{zy}w^j_tu^t_y
+ \Gamma^{ki}_lw^l_su^s_{zz}\delta_{zx}\nu_{zy}w^j_tu^t_y
\\
&+ \Gamma^{ki}_lw^l_su^s_z\delta_{zx}\delta_{zy}w^k_tu^t_y+
w^k_l g^{lj}_{,s}u^s_z\delta'_{zy}\nu_{zx}w^i_r\,u^r_x
\\
& + w^k_lg^{lj}\delta''_{zy}\nu_{zx}w^i_r\,u^r_x
  + w^k_l\Gamma^{lj}_{t,s}u^s_zu^t_z\delta_{zy}\nu_{zx}w^j_r\,u^r_x
\\
&+ w^k_l\Gamma^{lj}_tu^t_{zz}\delta_{zy}\nu_{zx}w^i_r\,u^r_x+
  w^k_l\Gamma^{lj}_tu^t_z\delta'_{zy}\nu_{zx}w^i_r\,u^r_x
\\
& + w^k_l w^l_{s,m}u^m_zu^s_z\nu_{zy}w^j_tu^t_y\nu_{zx}w^i_r\,u^r_x
+ w^k_l w^l_su^s_{zz}\nu_{zy}w^j_tu^t_y\nu_{zx}w^i_r\,u^r_x
\\
& + w^k_l w^l_su^s_z\delta_{zy}w^j_tu^t_y\nu_{zx}w^i_r\,u^r_x.
\end{split}
\end{equation}

The coefficients of the reduced form are listed below
\begin{itemize}
\item The coefficient of $\delta''_{xy}\delta_{xz}$ is
$g^{lj}\Gamma^{ki}_l$.
\item The coefficient of $\delta'_{xy}\delta_{xz}$ is
$\Gamma^{ki}_l g^{lj}_{,s}u^s_x+\Gamma^{ki}_l\Gamma^{lj}_tu^t_x$.
\item The coefficient of $\delta_{xy}\delta_{xz}$ is
  $ \Gamma^{ki}_l\Gamma^{lj}_{t,s}u^s_xu^t_x
  + \Gamma^{ki}_l\Gamma^{lj}_tu^t_{xx}
  +\Gamma^{ki}_lw^l_su^s_xw^j_tu^t_x$.
\item The coefficient of $\delta_{xz}\nu_{xy}$ is
  $\left(\Gamma^{ki}_l w^l_{s,m}u^m_xu^s_x
    +\Gamma^{ki}_lw^l_su^s_{xx}\right)w^j_tu^t_y$.
\item The coefficient of $\delta_{zy}\nu_{zx}$ is
  $\big(w^k_l\Gamma^{lj}_tu^t_{zz}
    + w^k_l\Gamma^{lj}_{t,s}u^s_zu^t_z
    +w^k_l w^l_su^s_zw^j_tu^t_z\big)w^j_r\,u^r_x$.
\item The coefficient of $\delta'_{zy}\nu_{zx}$ is
$w^k_l g^{lj}_{,s}u^s_zw^i_r\,u^r_x+ w^k_l\Gamma^{lj}_tu^t_zw^i_r\,u^r_x$.
\item The coefficient of $\nu_{zy}\nu_{zx}$ is
  $ w^k_l w^l_{s,m}u^m_zu^s_zw^j_tu^t_yw^i_r\,u^r_x
  +w^k_l w^l_su^s_{zz}w^j_tu^t_yw^i_r\,u^r_x$.
\item The coefficient of $\delta''_{zy}\nu_{zx}$ is
$w^k_lg^{lj}w^i_r\,u^r_x$.
\end{itemize}

The tenth summand in \eqref{SBDZ2} is
\begin{multline}
\frac{\partial P^{ki}_{z,x}}{\partial u^{l}_x} \partial_x P^{lj}_{x,y}=
w^k_su^s_z\nu_{zx}w^i_l g^{lj}_{,m}u^m_x\delta'_{xy}
+ w^k_su^s_z\nu_{zx}w^i_lg^{lj}\delta''_{xy}
\\
+ w^k_su^s_z\nu_{zx}w^i_l \Gamma^{lj}_{t,m}u^m_xu^t_x\delta_{xy}
+ w^k_su^s_z\nu_{zx}w^i_l\Gamma^{lj}_tu^t_{xx}\delta_{xy}+
w^k_su^s_z\nu_{zx}w^i_l\Gamma^{lj}_tu^t_x\delta'_{xy}
\\
+ w^k_su^s_z\nu_{zx}w^i_l w^l_{r,m}u^m_xu^r_x\nu_{xy}w^j_tu^t_y
+ w^k_su^s_z\nu_{zx}w^i_lw^l_ru^r_{xx}\nu_{xy}w^j_tu^t_y
\\
+ w^k_su^s_z\nu_{zx}w^i_lw^l_ru^r_x\delta_{xy}w^j_tu^t_y.
\end{multline}

The coefficients of the reduced form are listed below:
\begin{itemize}
\item The coefficient of $\delta_{xy}\delta_{xz}$ is
  $w^k_su^s_xw^i_l g^{lj}_{,m}u^m_x
  + w^k_su^s_xw^i_l\Gamma^{lj}_tu^t_x
  + \d_x(w^k_su^s_x)w^i_lg^{lj}$.
\item The coefficient of $\nu_{yz}\delta'_{yx}$ is
  $w^k_su^s_zw^i_l g^{lj}_{,m}u^m_y
  +w^k_su^s_zw^i_l\Gamma^{lj}_tu^t_y
-2w^k_su^s_z\d_y(w^i_lg^{lj})$.
\item The coefficient of $\nu_{zy}\delta_{yx}$ is
  \begin{multline*}
    - w^k_su^s_z\Big(\d_y(w^i_l g^{lj}_{,m}u^m_y)
      + \d_y^2(w^i_lg^{lj}) + w^i_l \Gamma^{lj}_{t,m}u^m_yu^t_y
      \\
      + w^i_l\Gamma^{lj}_tu^t_{yy}
      - \d_y(w^i_l\Gamma^{lj}_tu^t_{y})
      +w^i_lw^l_ru^r_yw^j_tu^t_y\Big).
\end{multline*}
\item The coefficient of $\nu_{zx}\nu_{xy}$ is
  $w^k_su^s_zw^i_l w^l_{r,m}u^m_xu^r_xw^j_tu^t_y
  + w^k_su^s_zw^i_lw^l_ru^r_{xx}w^j_tu^t_y$.
\item The coefficient of $\delta'_{xy}\delta_{xz}$ is
$2 w^k_su^s_xw^i_lg^{lj}$.
\item The coefficient of $\delta_{xy}\delta'_{xz}$ is
$w^k_su^s_xw^i_lg^{lj}$.
\item The coefficient of $\nu_{zy}\delta''_{yx}$ is
$w^k_su^s_zw^i_lg^{lj}$.
\end{itemize}

The eleventh summand in \eqref{SBDZ2} is
\begin{equation}
\begin{split}
\frac{\partial P^{jk}_{y,z}}{\partial u^{l}_y} & \partial_y P^{li}_{y,x}=
\Gamma^{jk}_l g^{li}_{,s}u^s_y\delta_{yz}\delta'_{yx}
+g^{li}\Gamma^{jk}_l\delta_{yz}\delta''_{yx}
+\Gamma^{jk}_l\Gamma^{li}_{t,s}u^s_yu^t_y\delta_{yz}\delta_{yx}
\\
&+\Gamma^{jk}_l\Gamma^{li}_tu^t_{yy}\delta_{yz}\delta_{yx}
+\Gamma^{jk}_l\Gamma^{li}_tu^t_y\delta_{yz}\delta'_{yx}
+ \Gamma^{jk}_l w^l_{s,m}u^m_yu^s_y\delta_{yz}\nu_{yx}w^i_tu^t_x
\\
& + \Gamma^{jk}_lw^l_su^s_{yy}\delta_{yz}\nu_{yx}w^i_tu^t_x
+ \Gamma^{jk}_lw^l_su^s_y\delta_{yz}\delta_{yx}w^i_tu^t_x
+ w^j_l g^{li}_{,s}u^s_y\delta'_{yx}\nu_{yz}w^k_r\,u^r_z
\\
& + w^j_lg^{li}\delta''_{yx}\nu_{yz}w^k_r\,u^r_z
+ w^j_l\Gamma^{li}_{t,s}u^s_yu^t_y\delta_{yx}\nu_{yz}w^k_r\,u^r_z
+ w^j_l\Gamma^{li}_tu^t_{yy}\delta_{yx}\nu_{yz}w^k_r\,u^r_z
\\
& + w^j_l\Gamma^{li}_tu^t_y\delta'_{yx}\nu_{yz}w^k_r\,u^r_z
+ w^j_lw^l_{s,m}u^m_yu^s_y\nu_{yx}w^i_tu^t_x\nu_{yz}w^k_r\,u^r_z
\\
&+ w^j_l w^l_su^s_{yy}\nu_{yx}w^i_tu^t_x\nu_{yz}w^k_r\,u^r_z+ w^j_l
w^l_su^s_y\delta_{yx}w^i_tu^t_x\nu_{yz}w^k_r\,u^r_z.
\end{split}\label{eq:566}
\end{equation}
The coefficients of the reduced form are listed below
\begin{itemize}
\item The coefficient of $\delta_{xy}\delta'_{xz}$ is
  $ - \Gamma^{jk}_l g^{li}_{,s}u^s_x+2\d_x(g^{li}\Gamma^{jk}_l)
  - \Gamma^{jk}_l\Gamma^{li}_tu^t_x$.
\item The coefficient of $\delta'_{xy}\delta_{xz}$ is
  $2\d_x(g^{li}\Gamma^{jk}_l) - \Gamma^{jk}_l g^{li}_{,s}u^s_x
  - \Gamma^{jk}_l\Gamma^{li}_tu^t_x$.
\item The coefficient of $\delta_{xy}\delta_{xz}$ is
  \begin{multline*}
    - \d_x(\Gamma^{jk}_l g^{li}_{,s}u^s_x)
    + \d_x^2(g^{li}\Gamma^{jk}_l)
    + \Gamma^{jk}_l\Gamma^{li}_{t,s}u^s_xu^t_x
    \\
    + \Gamma^{jk}_l\Gamma^{li}_tu^t_{xx}
    - \d_x(\Gamma^{jk}_l\Gamma^{li}_tu^t_x)
    + \Gamma^{jk}_lw^l_su^s_xw^i_tu^t_x.
  \end{multline*}
\item The coefficient of $\delta_{xy}\delta''_{xz}$ is
$g^{li}\Gamma^{jk}_l$.
\item The coefficient of $\delta'_{xy}\delta'_{xz}$ is
$2g^{li}\Gamma^{jk}_l$.
\item The coefficient of $\delta''_{xy}\delta_{xz}$ is
$g^{li}\Gamma^{jk}_l$.
\item The coefficient of $\delta_{yz}\nu_{zx}$ is
  $(\Gamma^{jk}_l w^l_{s,m}u^m_zu^s_z
    + \Gamma^{jk}_lw^l_su^s_{zz})w^i_tu^t_x$.
\item The coefficient of $\delta_{yx}\nu_{yz}$ is
  $(w^j_l\Gamma^{li}_{t,s}u^s_yu^t_y
  + w^j_l\Gamma^{li}_tu^t_{yy}+w^j_l w^l_su^s_yw^i_tu^t_y)
  w^k_r\,u^r_z$.
\item The coefficient of $\delta'_{yx}\nu_{yz}$ is
$w^j_l g^{li}_{,s}u^s_yw^k_r\,u^r_z
+ w^j_l\Gamma^{li}_tu^t_yw^k_r\,u^r_z$.
\item The coefficient of $\nu_{yx}\nu_{yz}$ is
  $w^j_l w^l_{s,m}u^m_yu^s_yw^i_tu^t_xw^k_r\,u^r_z
  +w^j_l w^l_su^s_{yy}w^i_tu^t_xw^k_r\,u^r_z$.
\item The coefficient of $\delta''_{yx}\nu_{yz}$ is
$w^j_lg^{li}w^k_r\,u^r_z$.
\end{itemize}

The twelfth (and last) summand in \eqref{SBDZ2} is
\begin{multline}
\frac{\partial P^{jk}_{y,z}}{\partial u^{l}_z} \partial_z P^{li}_{z,x}=
w^j_su^s_y\nu_{yz}w^k_l g^{li}_{,m}u^m_z\delta'_{zx}
+ w^j_su^s_y\nu_{yz}w^k_lg^{li}\delta''_{zx}
\\
+ w^j_su^s_y\nu_{yz}w^k_l \Gamma^{li}_{t,m}u^m_zu^t_z\delta_{zx}
+ w^j_su^s_y\nu_{yz}w^k_l\Gamma^{li}_tu^t_{zz}\delta_{zx}
+ w^j_su^s_y\nu_{yz}w^k_l\Gamma^{li}_tu^t_z\delta'_{zx}
\\
+ w^j_su^s_y\nu_{yz}w^k_l w^l_{r,m}u^m_zu^r_z\nu_{zx}w^i_tu^t_x+
w^j_su^s_y\nu_{yz}w^k_lw^l_ru^r_{zz}\nu_{zx}w^i_tu^t_x
\\
+ w^j_su^s_y\nu_{yz}w^k_lw^l_ru^r_z\delta_{zx}w^i_tu^t_x.
\end{multline}

The coefficients of the reduced form are listed below
\begin{itemize}
\item The coefficient of $\delta_{xy}\delta_{xz}$ is
  \begin{displaymath}
    w^j_su^s_xw^k_l g^{li}_{,m}u^m_x-\d_x(w^j_su^s_x)w^k_lg^{li}
    + w^j_su^s_xw^k_l\Gamma^{li}_tu^t_x-2 w^j_su^s_x\d_x(w^k_lg^{li}).
  \end{displaymath}
\item The  coefficient of $\nu_{yx}\delta''_{xz}$ is
$w^j_su^s_yw^k_lg^{li}$.
\item The coefficient of $\nu_{yx}\delta'_{xz}$ is
  $2w^j_su^s_y\d_x(w^k_lg^{li})-w^j_su^s_yw^k_l\Gamma^{li}_tu^t_x
  -w^j_su^s_yw^k_l g^{li}_{,m}u^m_x$.
\item The coefficient of $\delta_{xy}\delta'_{xz}$ is
$-2 w^j_su^s_xw^k_lg^{li}$.
\item The coefficient of $\delta'_{xy}\delta_{xz}$ is
$-w^j_su^s_xw^k_lg^{li}$.
\item The coefficient of $\nu_{yx}\delta_{xz}$ is
  \begin{multline*}
    w^j_su^s_y\Big(w^k_l \Gamma^{li}_{t,m}u^m_xu^t_x
      + w^k_l\Gamma^{li}_tu^t_{xx}
      + \d_x^2(w^k_lg^{li})
      \\
      - \d_x(w^k_l\Gamma^{li}_tu^t_x)
      + w^k_lw^l_ru^r_xw^i_tu^t_x
      - \d_x(w^k_l g^{li}_{,m}u^m_x)\Big).
\end{multline*}
\item The coefficient $\nu_{yz}\nu_{zx}$ is
  $w^j_su^s_yw^k_l w^l_{r,m}u^m_zu^r_zw^i_tu^t_x
  +w^j_su^s_yw^k_lw^l_ru^r_{zz}w^i_tu^t_x$.
\end{itemize}

\subsection{The conditions}
\label{sec:conditions}

Collecting all similar terms we get the following conditions
\begin{itemize}
\item The coefficients of $\delta''_{xy}\delta_{xz}$,
  $\delta'_{xy}\delta'_{xz}$ and $\delta_{xy}\delta''_{xz}$ vanish iff
  \eqref{eq:57} holds, namely $g^{li}\Gamma^{jk}_l=g^{lj}\Gamma^{ik}_l$.
  Combining this condition with skew-symmetry of the bracket~\eqref{eq:56} we
  obtain that $\Gamma^i_{jk}=-g_{jl}\Gamma^{li}_k$ are the Christoffel symbols
  of the Levi-Civita connection of $g$.
\item The coefficients of products of step functions vanish.
\item The coefficients of $\delta''_{xz}\nu_{xy}$, $\delta''_{zy}\nu_{zx}$ and
  $\delta''_{yx}\nu_{yz}$ vanish iff \eqref{eq:58} holds, namely
    $g_{ik}w^k_j=g_{jk}w^k_i$.
\item Using the above conditions the coefficient of
  $u^s_{xx}\delta_{xy}\delta_{xz}$ can be written as
  \begin{multline}\label{eq:21}
    g^{li}\Gamma^{jk}_{l,s} - g^{li} \Gamma^{jk}_{s,l}
    +\Gamma^{ij}_l\Gamma^{lk}_s-\Gamma^{ik}_l\Gamma^{lj}_s
    +g^{li}(w^k_sw^j_l-w^j_sw^k_l) =
    \\
    R^{ijk}_s+g^{li}(w^k_sw^j_l-w^j_sw^k_l),
\end{multline}
where $R^{ijk}_s$ is Riemann tensor (in upper indices). This yields the
condition \eqref{eq:60}.
\item  The coefficient of $u^t_{xx}\nu_{xy}\delta_{xz}$ (up to a common factor)
  is
\begin{equation} 
  g^{li}(w^k_{t,l}+\Gamma^{k}_{lm}w^m_t
    - w^k_{l,t}-\Gamma^{k}_{mt}w^m_l)
  = g^{li}(\nabla_lw^k_t-\nabla_tw^k_l),
\end{equation}
which yields the condition~\eqref{eq:59}.
\item The coefficient of $\delta'_{xy}\delta_{xz}$ is a linear combination of
  the coefficient~\eqref{eq:21} repeated two times.
\item The coefficients of $u^r_xu^s_{x}\delta_{xy}\delta_{xz}$ vanish using the
  $x$-derivative of the condition~\eqref{eq:21}.
  Indeed, the coefficient reduces to
\begin{align*}
  &( \Gamma^{ij}_{s,l}-\Gamma^{ij}_{l,s}+w^i_sw^j_l)\Gamma^{lk}_t
    +( \Gamma^{ki}_{s,l}-\Gamma^{ki}_{l,s}+w^k_sw^i_l)\Gamma^{lj}_t
  \\
  & +( \Gamma^{jk}_{s,l} - \Gamma^{jk}_{l,s}+w^j_sw^k_l)\Gamma^{li}_t
    + w^k_sg^{lj}( w^i_{t,l} - w^i_{l,t} + \Gamma^{i}_{lm}w^m_t)
  \\
  & + w^i_sg^{lk}( w^j_{t,l} - w^j_{l,t}+\Gamma^{j}_{lm}w^m_t)
  +w^j_sg^{li}( w^k_{t,l}- w^k_{l,t}+\Gamma^{k}_{lm}w^m_t).
\end{align*}
Using the condition~\eqref{eq:59} we obtain
\begin{align*}
  &( \Gamma^{ij}_{s,l}-\Gamma^{ij}_{l,s}+w^i_sw^j_l)\Gamma^{lk}_t
    +(\Gamma^{ki}_{s,l}- \Gamma^{ki}_{l,s}+w^k_sw^i_l)\Gamma^{lj}_t
  \\
  &+(\Gamma^{jk}_{s,l} - \Gamma^{jk}_{l,s} + w^j_sw^k_l)\Gamma^{li}_t+
    w^k_sg^{lm}(\Gamma^{i}_{tm}w^j_l)
  \\
  & +w^i_sg^{lm}(\Gamma^{j}_{tm}w^k_l) +w^j_sg^{lm}(\Gamma^{k}_{tm}w^i_l).
\end{align*}
Using again~\eqref{eq:21} we obtain
\begin{align*}
  &(\Gamma^{ij}_{s,l}-\Gamma^{ij}_{l,s}+w^i_sw^j_l-w^j_sw^i_l)\Gamma^{lk}_t
    +(\Gamma^{ki}_{s,l}- \Gamma^{ki}_{l,s}+w^k_sw^i_l-w^i_sw^k_l)\Gamma^{lj}_t+
  \\
  &(\Gamma^{jk}_{s,l}- \Gamma^{jk}_{l,s}+w^j_sw^k_l-w^k_sw^j_l)\Gamma^{li}_t=
    \Gamma^{lk}_t(-\Gamma^i_{lm}\Gamma^{mj}_s+\Gamma^j_{lm}\Gamma^{mi}_s)+
  \\
  & \Gamma^{lj}_t(-\Gamma^k_{lm}\Gamma^{mi}_s+\Gamma^i_{lm}\Gamma^{mk}_s)+
  \Gamma^{lj}_t(-\Gamma^j_{lm}\Gamma^{mk}_s+\Gamma^k_{lm}\Gamma^{mj}_s)=0
\end{align*}
\item The coefficient of $u^t_yu^r_{x}u^s_x\nu_{xy}\delta_{xz}$ vanishes
  due to the previous conditions. Indeed
\begin{align*}
  &u^s_xu^r_xw^j_tu^t_y\Big(w^l_s(\Gamma^{ki}_{r,l}-\Gamma^{ki}_{l,r}
    +w^i_lw^k_r-w^k_lw^i_r)
    -(\Gamma^{ki}_l+\Gamma^{ik}_l)\d_rw^l_s
  \\
  & +(\Gamma^{li}_r+\Gamma^{il}_r)(w^k_{s,l}-w^k_{l,s})
    -\Gamma^{ik}_{l,r}w^l_s+g^{li}(w^k_{s,l} - w^k_{l,s})_{,r}
  \\
  & + \Gamma^{li}_r(w^k_{l,s} - w^k_{s,l})
    +w^i_l\Gamma^{lk}_{s,r}+\Gamma^{ki}_lw^l_{s,r}
    +\Gamma^{lk}_s w^i_{r,l}\Big)
  \\
  = & u^s_xu^r_xw^j_tu^t_y\Big(w^l_s(\Gamma^{ki}_{r,l}-\Gamma^{ki}_{l,r}
  +w^i_lw^k_r-w^k_lw^i_r)
  -(\Gamma^{ki}_l+\Gamma^{ik}_l) w^l_{s,r}+
  \\
  & -\Gamma^{ik}_{l,r}w^l_s+
    (\Gamma^{ik}_{l}w^l_s-\Gamma^{lk}_{s}w^i_l)_{,r}
    + \Gamma^{li}_r(w^k_{l,s}-w^k_{s,l})
    \\
  & +w^i_l\Gamma^{lk}_{s,r}+\Gamma^{ki}_lw^l_{s,r}+\Gamma^{lk}_s
    w^i_{r,l}\Big)
  \\
  = &u^s_xu^r_xw^j_tu^t_y\Big(w^m_s(\Gamma^{k}_{rl}\Gamma^{li}_m
    -\Gamma^{i}_{rl}\Gamma^{lk}_m)
    -\Gamma^{lk}_{s}w^i_{l,r}
    \\
   & +\Gamma^{li}_r(\Gamma^{k}_{lm}w^m_s-\Gamma^{k}_{ms}w^m_l)
    +\Gamma^{lk}_s w^i_{r,l}\Big)
  \\
 = &u^s_xu^r_xw^j_tu^t_y\Big(w^m_s\Gamma^{k}_{rl}\Gamma^{li}_m
-\Gamma^{lk}_{s}( w^i_{l,r}+\Gamma^{i}_{mr}w^m_l)+\Gamma^{lk}_sw^i_{r,l}\Big)=0.
\end{align*} 
Similar computations hold for the coefficients of $\nu_{yz}\delta_{yx}$ and
$\nu_{xz}\delta_{zy}$.
\item The coefficient of $u^s_yu^t_x\delta'_{xz}\nu_{xy}$ is
\begin{multline*}
w^j_t\Big(g^{lk} w^i_{s,l}
- g^{ki}_{,l}w^l_s
+g^{li} w^k_{s,l}
+w^i_lg^{lk}_{,s}
\\
+w^i_l\Gamma^{lk}_s
-2(w^k_lg^{li})_{,s}
+w^k_l\Gamma^{li}_s 
+w^k_l g^{li}_{,s}\Big)
\\
=u^s_xu^t_y\Big(g^{lk}(\nabla_l w^i_s-\nabla_s w^i_l)
+g^{li}(\nabla_l w^k_s-\nabla_s w^k_l)\Big),
\end{multline*}
which vanishes upon~\eqref{eq:59}.
\end{itemize}

\section{Weakly nonlocal PBHT and differential operators}

Here we will just show the main steps of the algorithm in
Section~\ref{sec:algorithm-2}.

We assume that
\begin{equation}
  P = g^{ij}\ddx{} + \Gamma^{ij}_k u^k_x
  +  w^i_{k}u^k_x\ddx{-1}w^j_{ h}u^h_x\label{eq:12}
\end{equation}
where $\det (g^{ij})\neq 0$ and $\epsilon_\alpha\in\mathbb{R}$. 

We will compute the conditions of Hamiltonianity of the operator $P$ of the
type~\eqref{eq:12} using formula~\eqref{eq:22} and the
Algorithm~\ref{sec:algorithm-2}.  Let us set:
\begin{equation}
  \label{eq:29}
  P = L + N\quad\text{where}\quad
  L= g^{ij}\ddx{} + \Gamma^{ij}_k u^k_{x},\quad
  N= w^i_ku^k_{x}\ddx{-1}w^j_hu^h_{x}.
\end{equation}
The conditions of skew-adjointness are obvious.

\subsection{Calculation of the Jacobi identity}

From now on we will assume $L$ to be skew-adjoint ($N$ is skew-adjoint by
construction).
\begin{lemma}
  We have
  \begin{align}
    \notag
    \frac{1}{2}[P,P] = &\frac{1}{2}[L,L] + [L,N] + \frac{1}{2}[N,N]
                         \\ \label{eq:32} \begin{split}
    = & [\ell_{L,\psi^1}(L(\psi^2))(\psi^3) +
        \ell_{L,\psi^1}(N(\psi^2))(\psi^3) +
    \\
      &  \ell_{N,\psi^1}(L(\psi^2))(\psi^3) +
        \ell_{N,\psi^1}(N(\psi^2))(\psi^3) +
        \text{cyclic}(\psi^1,\psi^2,\psi^3)]
        \end{split}
  \end{align}
\end{lemma}

We begin by computing the linearization of $L$ and $N$.  Let us
introduce the new non-local scalar functions
\begin{equation}
  \label{eq:71}
  \tilde{\psi}^k = \ddx{-1}(w^i_l u^l_{x}\psi^k_i),\qquad k=1,2,3.
\end{equation}

\begin{lemma}
The linearization of $L$ and $N$ have the following expressions:
\begin{align}
  \label{eq:43}
  \ell_{L,\psi^1}(\varphi)^i = &\Big(g^{ij}_{,k}\ddx{}\psi^1_j
  + \Gamma^{ij}_{h,k}u^h_{x}\psi^1_j\Big)\varphi^k
  + \Gamma^{ij}_{h}\psi^1_j \ddx{}\varphi^h
  \\ \notag
  \ell_{N,\psi^1}(\varphi)^i = &
  (w^i_{k,l}u^k_{x}\varphi^l + w^i_{k}\ddx{}\varphi^k)
                           \tilde{\psi}^1 % \ddx{-1}(w^j_hb^h_{2x}\psi^1_j)
  \\ \label{eq:62}
  & + w^i_ku^k_{x}\ddx{-1}\big( (w^j_{h,l}u^h_{x}\varphi^l +
  w^j_l\ddx{}\varphi^l)\psi^1_j \big)
\end{align}
\end{lemma}
We compute the first summand of~\eqref{eq:32}:
\begin{equation}
  \label{eq:31}
  \begin{split}
  \ell_{L,\psi^1}&(L(\psi^2))(\psi^3) =
   g^{ij}_{,k}g^{kp}\ddx{}\psi^1_j\ddx{}\psi^2_p\psi^3_i
   + g^{ij}_{,k}\Gamma^{kp}_h u^h_{x}\ddx{}\psi^1_j\psi^2_p\psi^3_i
  \\ %010
  & + \Big(\Gamma^{ij}_{h,k}u^h_{x}g^{kp} + \Gamma^{ij}_{h} \ddx{}(g^{hp})
     + \Gamma^{ij}_{h}\Gamma^{hp}_{k} u^k_{x}\Big)
     \psi^1_j\ddx{}\psi^2_p\psi^3_i
  \\ %020
  & + \Gamma^{ij}_{h} g^{hp}\psi^1_j\ddx{2}\psi^2_p\psi^3_i
   + \Big(\Gamma^{ij}_{h,k}u^h_{x}\Gamma^{kp}_h u^h_{x}
      + \Gamma^{ij}_{h}\ddx{}(\Gamma^{hp}_{k}u^k_{x})\Big)
      \psi^1_j\psi^2_p\psi^3_i.
      \end{split}
\end{equation}
We compute the second summand of~\eqref{eq:32}:
\begin{multline}
  \label{eq:33}
  \ell_{L,\psi^1}(N(\psi^2))(\psi^3)=
   g^{ij}_{,k} w^k_mu^m_{x}\ddx{}\psi^1_j\tilde{\psi}^2\psi^3_i
\\  + \Big(\Gamma^{ij}_{h,k}u^h_{x}w^k_mu^m_{x}
     + \Gamma^{ij}_{k}\ddx{}(w^k_hu^h_{x})\Big)
     \psi^1_j\tilde{\psi}^2\psi^3_i
   + \Gamma^{ij}_{k} w^k_hu^h_{x}w^p_lu^l_{x}
                   \psi^1_j\psi^2_p\psi^3_i.
\end{multline}
We compute the third summand of~\eqref{eq:32}. Here we integrate non-local
expressions by parts in order to concentrate integrals in expressions of the
form~\eqref{eq:71}:
\begin{equation}
  \label{eq:34}
  \begin{split}
  \ell_{N,\psi^1}&(L(\psi^2))(\psi^3) =
    (w^i_{l,k}u^l_{x}g^{kp}+ w^i_{k}\ddx{}(g^{kp})
      + w^i_{k}\Gamma^{kp}_m u^m_{x})
      \tilde{\psi}^1\ddx{}\psi^2_p\psi^3_i
  \\
  & %-100
  + (w^i_{l,k}u^l_{x}\Gamma^{kp}_m u^m_{x}+ w^i_{k}\ddx{}(\Gamma^{kp}_m u^m_{x}))
  \tilde{\psi}^1\psi^2_p\psi^3_i
  +w^i_{k}g^{kp}\tilde{\psi}^1\ddx{2}\psi^2_p\psi^3_i
  \\ %01-1
  & + (- w^j_{h,k}u^h_{x}g^{kp} - w^j_k\ddx{}(g^{kp})
                   - w^j_k\Gamma^{kp}_m u^m_{x})
     \psi^1_j\ddx{}\psi^2_p\tilde{\psi}^3
  \\ %00-1
  & + (- w^j_{h,k}u^h_{x}\Gamma^{kp}_m u^m_{x}
    - w^j_k\ddx{}(\Gamma^{kp}_m u^m_{x}))
    \psi^1_j\psi^2_p\tilde{\psi}^3
   - w^j_kg^{kp}\psi^1_j\ddx{2}\psi^2_p\tilde{\psi}^3.
\end{split}
\end{equation}
We compute the fourth summand of~\eqref{eq:32}. Here we integrate non-local
expressions by parts in order to concentrate integrals in expressions of the
form~\eqref{eq:71}:
\begin{equation}
  \label{eq:35}
  \begin{split}
  \ell_{N,\psi^1}&(N(\psi^2))(\psi^3) =
   (w^i_{k,l}u^k_{x}w^l_hu^h_{x} + w^i_{k}\ddx{}(w^k_mu^m_{x}))
     \tilde{\psi}^1\tilde{\psi}^2\psi^3_i
  \\ %-100
  & + w^i_{k}w^k_mu^m_{x}w^p_hu^h_{x}\tilde{\psi}^1\psi^2_p\psi^3_i
   +(- w^j_{h,l}u^h_{x}w^l_ku^k_{x}- w^j_l\ddx{}(w^l_mu^m_{x}))
                   \psi^1_j \tilde{\psi}^2\tilde{\psi}^3
  \\
  & - w^j_lw^l_mu^m_{x}w^p_hu^h_{x} \psi^1_j\psi^2_p\tilde{\psi}^3.
\end{split}
\end{equation}
The three-vector~\eqref{eq:32} can be written, after adding the cyclically
permuted summands, as the sum $T_l+T_n$, where $T_l$ is the local part and
$T_n$ is the non-local part.

\subsection{Calculation of the reduced form}

Now, we fix the indices $1$, $2$, $3$ and we bring the nonlocal part to the
normal form with respect to the three ordered indices. This means that the
terms which are quadratic in the nonlocal expressions~\eqref{eq:71} shall be
preserved, while terms which are linear in the nonlocal expressions should be
brought to one of the following forms by integrating by parts:
\begin{equation}
  \label{eq:38}
  \tilde{\psi}^1\ddx{k}\psi^2_p\psi^3_i,\quad % form `a'
  \tilde{\psi}^2\ddx{k}\psi^3_p\psi^1_i,\quad % form `b'
  \tilde{\psi}^3\ddx{k}\psi^1_p\psi^2_i. % form `c'
\end{equation}
For example,
$g^{ij}_{,k} w^k_mu^m_{x} \tilde{\psi}^2\psi^3_i\ddx{}\psi^1_j$ must be
replaced by
$ - \ddx{}\big(g^{ij}_{,k} w^k_mu^m_{x} \tilde{\psi}^2\psi^3_i\big)\psi^1_j$
(of course, up to a total divergence). After the above computational step we
can write the final form of the nonlocal part of the three-vector. Note that
$T_n$ acquired some local terms at the end of the first step of the
algorithm. We introduce the notation $T_n=T_N + T_{nL}$, where $T_N$ is the
non-local part and $T_{nL}$ is the local part of $T_n$ after the first step of
the algorithm. We have, after collecting like terms:
\begin{align}
  \label{eq:42}
  T_N=&\notag
        \\
  \begin{split}
  & \Big(- \ddx{}\big(g^{ij}_{,k} w^k_mu^m_{x}\big)
          + \Gamma^{ij}_{h,k}u^h_{x}w^k_mu^m_{x}
    + \Gamma^{ij}_{k}\ddx{}(w^k_hu^h_{x})
    \\
    & \hphantom{ciao}
    + w^j_{l,k}u^l_{x}\Gamma^{ki}_m u^m_{x}
    + w^j_{k}\ddx{}(\Gamma^{ki}_m u^m_{x})\Big)
    \\
    & \hphantom{ciao}
    -\ddx{}\big( - w^i_{h,k}u^h_{x}g^{kj} - w^i_k\ddx{}(g^{kj})
    - w^i_k\Gamma^{kj}_m u^m_{x}\big)
    \\
    & \hphantom{ciao}
    - w^i_{h,k}u^h_{x}\Gamma^{kj}_m u^m_{x}
    - w^i_k\ddx{}(\Gamma^{kj}_m u^m_{x}) + \ddx{2}\big(- w^i_kg^{kj}\big)
    \\
    & \hphantom{ciao}
    + w^j_{k}w^k_mu^m_{x}w^i_hu^h_{x}
    - w^i_lw^l_mu^m_{x}w^j_hu^h_{x} \Big)
    \tilde{\psi}^2\psi^3_i\psi^1_j
  \end{split}
  \\ %b-110
  \begin{split}
  & \Big(- g^{ij}_{,k} w^k_mu^m_{x} +
     w^j_{l,k}u^l_{x}g^{ki}+ w^j_{k}\ddx{}(g^{ki})
     + w^j_{k}\Gamma^{ki}_m u^m_{x} 
  \\
  &\hphantom{ciao}
  + w^i_{h,k}u^h_{x}g^{kj} + w^i_k\ddx{}(g^{kj})
  + w^i_k\Gamma^{kj}_m u^m_{x}
  \\
  &\hphantom{ciao}
  + 2\ddx{}\big(- w^i_kg^{kj}\big)\Big)
           \tilde{\psi}^2\ddx{}\psi^3_i\psi^1_j
  \end{split}
  \\
  & %b-120
  + (w^j_{k}g^{ki}- w^i_kg^{kj})\tilde{\psi}^2\ddx{2}\psi^3_i\psi^1_j
\end{align}
plus a cyclic permutation of the above terms.
Moreover:
\begin{align}
  \label{eq:45}
  T_{nL} =&\notag
            \\
  \begin{split}
  & \Big(- g^{ij}_{,k} w^k_mu^m_{x} w^p_lu^l_x  - g^{jp}_{,k} w^k_mu^m_{x} w^i_lu^l_x 
    - g^{pi}_{,k} w^k_mu^m_{x} w^j_lu^l_x
  \\
  &\hphantom{ciao}
  + (w^j_{h,k}u^h_{x}g^{kp} + w^j_k\ddx{}(g^{kp})
    + w^j_k\Gamma^{kp}_m u^m_{x}) w^i_lu^l_x
  \\
  &\hphantom{ciao}
    - 2\ddx{}\big(w^j_kg^{kp}\big)w^i_lu^l_x - w^j_kg^{kp}\ddx{}(w^i_lu^l_x)
  \\
  &\hphantom{ciao}
  + (w^p_{h,k}u^h_{x}g^{ki} + w^p_k\ddx{}(g^{ki})
    + w^p_k\Gamma^{ki}_m u^m_{x})w^j_lu^l_x
  \\
  &\hphantom{ciao}
    + 2 \ddx{}\big(- w^p_kg^{ki}\big)w^j_lu^l_x
    + (- w^p_kg^{ki})\ddx{}(w^j_lu^l_x)
  \\
  &\hphantom{ciao}
  + (w^i_{h,k}u^h_{x}g^{kj} + w^i_k\ddx{}(g^{kj})
    + w^i_k\Gamma^{kj}_m u^m_{x})w^p_lu^l_x
  \\
  &\hphantom{ciao}
  + (- w^i_kg^{kj})\ddx{}(w^p_lu^l_x) + 2\ddx{}\big(- w^i_kg^{kj}\big)w^p_lu^l_x
  \Big)
  \\
  &\hphantom{ciao}
  \psi^1_j\psi^2_p\psi^3_i
\end{split}
  \\ %100
          & \big(- 2 w^j_kg^{kp}w^i_lu^l_x - w^p_kg^{ki}w^j_lu^l_x\big)
            \ddx{}\psi^1_j\psi^2_p\psi^3_i
\end{align}
plus a cyclic permutation of the last summand.

Let us introduce the notation $T_L=T_l + T_{nL}$. We shall bring each summand
of $T_L$ to the canonical form
\begin{equation}
  \label{eq:46}
  c^{jpi}\ddx{k}\psi^1_j\ddx{h}\psi^2_p \psi^3_i,
\end{equation}
where $c^{jpi}$ are coefficient functions, using integration by parts on
summands that contain $\ddx{l}\psi^3_i$ with $l>0$. We have:
\begin{align*}
  \label{eq:367}
  T_L = &
  \\ %110
     & \Big(g^{ij}_{,k}g^{kp} - g^{jp}_{,k}g^{ki}
          + 2\Gamma^{jp}_{h} g^{hi} - g^{pi}_{,k}g^{kj}\Big)
       \ddx{}\psi^1_j\ddx{}\psi^2_p\psi^3_i
  \\ %100
%  \begin{split}
        & + \Big( g^{ij}_{,k}\Gamma^{kp}_h u^h_{x}
          - \Gamma^{jp}_{h,k}u^h_{x}g^{ki}
          - \Gamma^{jp}_{h} \ddx{}(g^{hi})
          - \Gamma^{jp}_{h}\Gamma^{hi}_{k} u^k_{x}
          \\
          &\hphantom{ciao}
          + 2\ddx{}\big(\Gamma^{jp}_{h} g^{hi}\big)
          -\ddx{}\big(g^{pi}_{,k}g^{kj}\big)
          - g^{pi}_{,k}\Gamma^{kj}_h u^h_{x}
          + \Gamma^{pi}_{h,k}u^h_{x}g^{kj}
          \\
          &\hphantom{ciao}
          + \Gamma^{pi}_{h} \ddx{}(g^{hj})
          + \Gamma^{pi}_{h}\Gamma^{hj}_{k} u^k_{x}
          \\
          &\hphantom{ciao}
          - 2 w^j_kg^{kp}w^i_lu^l_x - w^p_kg^{ki}w^j_lu^l_x
          + w^j_kg^{kp}w^i_lu^l_x
          + 2w^i_kg^{kj}w^p_lu^l_x\Big)
          \\
          &\hphantom{ciao}
          \ddx{}\psi^1_j\psi^2_p\psi^3_i
%  \end{split}
  \\ %010
%  \begin{split}
  & + \Big(
    \Gamma^{ij}_{h,k}u^h_{x}g^{kp} + \Gamma^{ij}_{h} \ddx{}(g^{hp})
    + \Gamma^{ij}_{h}\Gamma^{hp}_{k} u^k_{x}
    - \ddx{}\Big(g^{jp}_{,k}g^{ki}\Big)
    \\
    & \hphantom{ciao}
    + g^{jp}_{,k}\Gamma^{ki}_h u^h_{x}
    - \Gamma^{jp}_{h,k}u^h_{x}g^{ki}
    - \Gamma^{jp}_{h} \ddx{}(g^{hi})
    - \Gamma^{jp}_{h}\Gamma^{hi}_{k} u^k_{x}
    \\
    & \hphantom{ciao}
    + 2\ddx{}\big(\Gamma^{jp}_{h} g^{hi}\big)
    - g^{pi}_{,k}\Gamma^{kj}_h u^h_{x}
    \\
    & \hphantom{ciao}
     + w^j_kg^{kp}w^i_lu^l_x
    + 2w^i_kg^{kj}w^p_lu^l_x
    - 2w^p_kg^{ki}w^j_lu^i_x
    - w^i_kg^{kj}w^p_lu^l_x
    \Big)
    \\
    & \hphantom{ciao}
    \psi^1_j\ddx{}\psi^2_p\psi^3_i
%  \end{split}
  \\ %020
        & + \Big(\Gamma^{ij}_{h} g^{hp} - g^{jp}_{,k}g^{ki}
          + \Gamma^{jp}_{h} g^{hi}\Big)
          \psi^1_j\ddx{2}\psi^2_p\psi^3_i
  \\ %000
%  \begin{split}
    & + \Big(
    \Gamma^{ij}_{h,k}u^h_{x}\Gamma^{kp}_l u^l_{x}
    + \Gamma^{ij}_{h}\ddx{}(\Gamma^{hp}_{k}u^k_{x})
    \\
    &\hphantom{ciao}
    -\ddx{}\big(\Gamma^{jp}_{h,k}u^h_{x}g^{ki}
          + \Gamma^{jp}_{h} \ddx{}(g^{hi})
    + \Gamma^{jp}_{h}\Gamma^{hi}_{k} u^k_{x}\big)
    + \ddx{2}\big(\Gamma^{jp}_{h} g^{hi}\big)
    \\
    &\hphantom{ciao}
    + \Gamma^{jp}_{h,k}u^h_{x}\Gamma^{ki}_l u^l_{x}
    + \Gamma^{jp}_{h}\ddx{}(\Gamma^{hi}_{k}u^k_{x})
    -\ddx{}\big(g^{pi}_{,k}\Gamma^{kj}_h u^h_{x}\big)
    \\
    &\hphantom{ciao}
    + \Gamma^{pi}_{h,k}u^h_{x}\Gamma^{kj}_l u^l_{x}
    + \Gamma^{pi}_{h}\ddx{}(\Gamma^{hj}_{k}u^k_{x})
    \\
    &\hphantom{ciao}
    + \Gamma^{ij}_{k} w^k_hu^h_{x}w^p_lu^l_{x}
    + \Gamma^{jp}_{k} w^k_hu^h_{x}w^i_lu^l_{x}
    + \Gamma^{pi}_{k} w^k_hu^h_{x}w^j_lu^l_{x}
    \\
    &\hphantom{ciao}
    - g^{ij}_{,k} w^k_mu^m_{x} w^p_lu^l_x
    - g^{jp}_{,k} w^k_mu^m_{x} w^i_lu^l_x
    - g^{pi}_{,k} w^k_mu^m_{x} w^j_lu^l_x
  \\
  &\hphantom{ciao}
  + (w^j_{h,k}u^h_{x}g^{kp} + w^j_k\ddx{}(g^{kp})
    + w^j_k\Gamma^{kp}_m u^m_{x}) w^i_lu^l_x
  \\
  &\hphantom{ciao}
    - 2\ddx{}\big(w^j_kg^{kp}\big)w^i_lu^l_x - w^j_kg^{kp}\ddx{}(w^i_lu^l_x)
  \\
  &\hphantom{ciao}
  + (w^p_{h,k}u^h_{x}g^{ki} + w^p_k\ddx{}(g^{ki})
    + w^p_k\Gamma^{ki}_m u^m_{x})w^j_lu^l_x
  \\
  &\hphantom{ciao}
    + 2 \ddx{}\big(- w^p_kg^{ki}\big)w^j_lu^l_x
    + (- w^p_kg^{ki})\ddx{}(w^j_lu^l_x)
  \\
  &\hphantom{ciao}
  + (w^i_{h,k}u^h_{x}g^{kj} + w^i_k\ddx{}(g^{kj})
    + w^i_k\Gamma^{kj}_m u^m_{x})w^p_lu^l_x
  \\
  &\hphantom{ciao}
  + (- w^i_kg^{kj})\ddx{}(w^p_lu^l_x) + 2\ddx{}\big(- w^i_kg^{kj}\big)w^p_lu^l_x
  \\
  &\hphantom{ciao}
    - \ddx{}\Big(- w^j_kg^{kp}w^i_lu^l_x
          - 2w^i_kg^{kj}w^p_lu^l_x\Big)
  \\
  &\hphantom{ciao}
          \psi^1_j\psi^2_p\psi^3_i
%\end{split}
   \\ %200
        & + \Big(\Gamma^{jp}_{h} g^{hi}- g^{pi}_{,k}g^{kj}+
          \Gamma^{pi}_{h}g^{hj}
          \Big)
          \ddx{2}\psi^1_j\psi^2_p\psi^3_i
\end{align*}

\subsection{The conditions}
\label{sec:conditions-1}

The vanishing of coefficients of the $3$-vector $T$ yields the conditions on
$P$ to be Hamiltonian. Below we list the basic elements of $T$ and the
conditions that arise from their coefficients. We assume the condition of
skew-adjointness of $P$.
\begin{description}
\item[$\psi^1_j\ddx{2}\psi^2_p\psi^3_i$:] the coefficient is
  \begin{equation}
    \label{eq:47}
    \Gamma^{ij}_{h} g^{hp} - g^{jp}_{,k}g^{ki} + \Gamma^{jp}_{h} g^{hi}
  \end{equation}
  and corresponds to the coefficient of $\delta''_{xy}\delta_{xz}$ and similar
  terms in Section~\ref{sec:conditions}. Its vanishing is equivalent to the
  condition $\Gamma^{jp}_{h} g^{hi} = \Gamma^{ip}_{h} g^{hj}$.
\item[$\ddx{}\psi^1_j\ddx{}\psi^2_p\psi^3_i$:] the coefficient vanish on
  account of the above condition.
\item[$\tilde{\psi}^1\ddx{2}\psi^2_p\psi^3_i$:] the coefficient is
  \begin{equation}
    \label{eq:49}
    w^i_kg^{kp} - w^p_k g^{ki}
  \end{equation}
  and corresponds to the coefficients of $\delta''_{xz}\nu_{xy}$ and similar
  terms in Section~\ref{sec:conditions}.
\item[$\psi^1_j\psi^2_p\psi^3_i$:] This coefficient is a differential
  polynomial; the coefficient of $u^k_{xx}$ reduces to
  \begin{equation}
    \label{eq:50}
    (\Gamma^{jp}_{h,k} - \Gamma^{jp}_{k,h})g^{hi} + \Gamma^{ij}_h\Gamma^{hp}_k
    - \Gamma^{ip}_h\Gamma^{hj}_k + g^{hi}(w^j_hw^p_k - w^p_hw^j_k)
  \end{equation}
  using \eqref{eq:49}. This corresponds to the coefficient of
  $u^k_{xx}\delta_{xy}\delta_{xz}$ in Section~\ref{sec:conditions}.
\item[$\tilde{\psi}^2\psi^3_i\psi^1_j$:] This coefficient is a differential
  polynomial; the coefficient of $u^m_{xx}$ reduces to
  \begin{equation}
    \label{eq:51}
    -g^{ij}_{,k}w^k_m + \Gamma^{ij}_kw^k_m +w^j_k\Gamma^{ki}_m
    + w^i_{m,k}g^{kj} - w^i_{k,m}g^{kj},
  \end{equation}
  and corresponds to the coefficient of $u^m_{xx}\nu_{xy}\delta_{xz}$ in
  Section~\ref{sec:conditions}.  The coefficient is equal to
  $g^{kj}(\nabla_k w^i_m - \nabla_m w^i_k)$.
\end{description}
The correspondence between the coefficients of the three-vector in the language
of operators and of distributions extends to all remaining terms; there is no
need to repeat the computation that shows that all other coefficients vanish on
account of the above conditions.

\section{Weakly nonlocal PBHT and Poisson Vertex Algebras}

In this section we will use the master formula to compute the skewsymmetry
condition and the PVA-Jacobi identity for the $\lambda$ bracket
\begin{equation}\label{eq:def}
\{u^i_{\lambda}u^j\}_P= g^{ji}\lambda + \Gamma^{ji}_s u^s_x
+ w^j_mu^m_x(\lambda+\d)^{-1}w^i_nu^n_x,
\end{equation}
corresponding to the weakly non-local Hamiltonian operator \eqref{eq:29}. As
before, we split the operator in the local and nonlocal parts
\begin{align}\label{eq:lambdaL}
\{u^i_{\lambda}u^j\}_L&= g^{ji}\lambda + \Gamma^{ji}_s u^s_x,
\\
\label{eq:lambdaN}
\{u^i_{\lambda}u^j\}_N&= w^j_mu^m_x(\lambda+\d)^{-1}w^i_nu^n_x.
\end{align}

Enforcing Property 5 of Section \ref{sec:nPVA1} on the two $\lambda$ brackets \eqref{eq:lambdaL} and
\eqref{eq:lambdaN} gives the conditions for the corresponding operator to be
skewsymmetric.

Indeed, for the local part, we have
\begin{equation}
\label{eq:22}
\{u^i_\lambda u^j\}^L=g^{ji}\lambda+\Gamma^{ji}_su^s_{x}=-{}_\to\{u^j_{-\lambda-\d}u^i\}^L=g^{ij}\lambda+\d_s g^{ij}u^s_{x}-\Gamma^{ij}_su^s_{x},
\end{equation}
which implies the conditions~\eqref{eq:52}, \eqref{eq:56}, and it is easy to
prove that the nonlocal part is skewsymmetric by construction.

\subsection{Computations with the master formula}
For convenience, we split the PVA-Jacobi identity on the generators -- defined in Section \ref{sec:PVA-voc} -- in the four parts
\begin{equation*}
J^{ijk}_{\lambda,\mu}(P,P)=J^{ijk}_{\lambda,\mu}(L,L)+J^{ijk}_{\lambda,\mu}(N,N)+J^{ijk}_{\lambda,\mu}(N,L)+J^{ijk}_{\lambda,\mu}(L,N),
\end{equation*}
where the last two terms correspond to the Schouten bracket $[L,N]$.

The purely local part $J^{ijk}_{\lambda,\mu}(L,L)$ is a straightforward application of the master formula:
\begin{equation}\label{eq:LL1}
\begin{split}
\{u^i_{\lambda}\{u^j_\mu u^k\}^L\}^L&=g^{li}\d_l g^{kj}\lambda\mu+\d_lg^{kj}\Gamma^{li}_su^s_x\mu+g^{li}\d_l\Gamma^{kj}_su^s_x\lambda\\
&+\Gamma^{li}_s\d_l\Gamma^{kj}_tu^s_xu^t_x+g^{li}\Gamma^{kj}_l\lambda^2+\d_sg^{li}\Gamma^{kj}_lu^s_x\lambda\\
&+\Gamma^{kj}_l\Gamma^{li}_{s}u^s_x\lambda+\Gamma^{kj}_l\d_s\Gamma^{li}_tu^t_xu^s_x+\Gamma^{kj}_l\Gamma^{li}_su^s_{xx}
\end{split}
\end{equation}
\begin{equation}\label{eq:LL2}
\begin{split}
\{u^j_{\mu}\{u^i_\mu u^k\}^L\}^L&=g^{lj}\d_l g^{ki}\lambda\mu+\d_lg^{ki}\Gamma^{lj}_su^s_x\lambda+g^{lj}\d_l\Gamma^{ki}_su^s_x\mu\\
&+\Gamma^{lj}_s\d_l\Gamma^{ki}_tu^s_xu^t_x+g^{lj}\Gamma^{ki}_l\mu^2+\d_sg^{lj}\Gamma^{ki}_lu^s_x\mu\\
&+\Gamma^{ki}_l\Gamma^{lj}_{s}u^s_x\mu+\Gamma^{ki}_l\d_s\Gamma^{lj}_tu^t_xu^s_x+\Gamma^{ki}_l\Gamma^{lj}_su^s_{xx}
\end{split}
\end{equation}
\begin{equation}\label{eq:LL3}
\begin{split}
\{\{u^i_{\lambda}u^j\}^L_{\lambda+\mu}u^k\}^L&=g^{kl}\d_lg^{ji}\lambda^2+g^{kl}\d_lg^{ji}\lambda\mu+g^{kl}\d_{sl}g^{ji}u^s_x\lambda\\
&+g^{kl}\d_l\Gamma^{ji}_su^s_x\lambda+g^{kl}\d_l\Gamma^{ji}_su^s_x\mu+g^{kl}\d_{sl}\Gamma^{ji}_tu^s_xu^t_x\\
&+g^{kl}\d_{l}\Gamma^{ji}_su^s_{xx}-g^{kl}\Gamma^{ji}_l\lambda^2-g^{kl}\Gamma^{ji}_l\mu^2\\
&-2g^{kl}\Gamma^{ji}_l\lambda\mu-2g^{kl}\d_s\Gamma^{ji}_lu^s_x\lambda-2g^{kl}\d_s\Gamma^{ji}_lu^s_x\mu\\
&-g^{kl}\d_s\Gamma^{ji}_lu^s_{xx}-g^{kl}\d_{st}\Gamma^{ji}_lu^s_xu^t_x+\Gamma^{kl}_s\d_lg^{ji}u^s_x\lambda\\
&+\Gamma^{kl}_s\d_l\Gamma^{ji}_tu^s_xu^t_x-\Gamma^{kl}_s\Gamma^{ji}_lu^s_x\lambda-\Gamma^{kl}_s\Gamma^{ji}_lu^s_x\mu\\
&-\Gamma^{kl}_s\d_t\Gamma^{ji}_lu^s_xu^t_x
\end{split}
\end{equation}
where all the monomials are of the form $\lambda^p\mu^q$ with $p,q\geq0$.

Computing the expressions with nonlocal terms is more complicated. However, it is possible to rely on Leibniz's and sesquilinearity properties of the $\lambda$ brackets to split the problem into smaller chunks. The basic observation, that can be proved by expanding the $(\lambda+\d)^{-1}$ expression, is that
$$
\{f_{\mu}(\lambda+\d)^{-1}g\}=(\lambda+\mu+\d)^{-1}\{f_{\mu}g\}.
$$
Let us start with $J^{ijk}_{\lambda,\mu}(N,L)$. Combining this with the left
and right Leibnitz properties we get
\begin{align}\label{eq:NL1}
\{u^i_{\lambda}&\{u^j_\mu u^k\}^N\}^L=\{u^i_\lambda
                 w^k_mu^m_x(\mu+\d)^{-1}w^j_nu^n_x\}^L \notag
\\
&=\{u^i_\lambda w^k_mu^m_x\}^L[(\mu+\d)^{-1}w^j_nu^n_x]+w^k_mu^m_x\{u^i_\lambda
     (\mu+\d)^{-1}w^j_nu^n_x\}^L\notag
  \\
&=\{u^i_\lambda w^k_mu^m_x\}^L[(\mu+\d)^{-1}w^j_nu^n_x]
  +w^k_mu^m_x(\lambda+\mu+\d)^{-1}\{u^i_\lambda w^j_nu^n_x\}^L\notag
  \\
%    \begin{split}
               &=[(\mu+\d)^{-1}w^j_nu^n_x]\left(g^{li}w^k_l\lambda^2\right.
                 \notag
  \\
  &\qquad+g^{li}\d_lw^k_su^s_x\lambda+w^k_l\d_sg^{li}u^s_x\lambda
  +w^k_l\Gamma^{li}_su^s_x\lambda\notag
  \\
  &\qquad+\left.\Gamma^{li}_s\d_lw^k_tu^s_xu^t_x
  +w^k_l\d_t\Gamma^{li}_su^s_xu^t_x+w^k_l\Gamma^{li}_su^s_{xx}\right)\notag
  \\
  &+w^k_mu^m_x(\lambda+\mu+\d)^{-1}\left(w^j_lg^{li}\lambda^2\right.\notag
  \\
  &\qquad+g^{li}\d_lw^j_su^s_x\lambda+w^j_l\d_sg^{li}u^s_x\lambda
    +w^j_l\Gamma^{li}_su^s_x\lambda\notag
  \\
  &\qquad\left.+\d_lw^j_s\Gamma^{li}_tu^s_xu^t_x
  +w^j_l\d_t\Gamma^{li}_su^s_xu^t_x+w^j_l\Gamma^{li}_su^s_{xx}\right),
%\end{split}
\\
  \label{eq:NL2}
%  \begin{split}
\{u^j_{\lambda}&\{u^i_\lambda u^k\}^N\}^L=
   [(\lambda+\d)^{-1}w^i_nu^n_x]\left(g^{lj}w^k_l\mu^2\right.\notag
  \\
               &\qquad+g^{lj}\d_lw^k_su^s_x\mu+w^k_l\d_sg^{lj}u^s_x\mu
                 +w^k_l\Gamma^{lj}_su^s_x\mu\notag
  \\
&\qquad\left.+\Gamma^{lj}_s\d_lw^k_tu^s_xu^t_x+w^k_l\d_t\Gamma^{lj}_su^s_xu^t_x+
       w^k_l\Gamma^{lj}_su^s_{xx}\right)\notag
  \\
  &+w^k_mu^m_x(\lambda+\mu+\d)^{-1}\left(\d_lw^i_s\Gamma^{lj}_tu^s_xu^t_x
      \right.\notag
  \\
&\qquad\left.+g^{lj}\d_lw^i_su^s_x\mu+w^i_l(\mu+\d)\left(g^{lj}\mu+\Gamma^{lj}_su^s_x\right)\right),
%\end{split}
\\
\label{eq:NL3}
  \{\{u^i_{\lambda}&u^j\}^N_{\lambda+\mu}u^k\}^L=
   \{w^j_mu^m_x(\lambda+\d)^{-1}w^i_nu^n_x{}_{\lambda+\mu}u^k\}^L\notag
  \\
  &=\{w^j_mu^m_x{}_{\lambda+\mu+\d}u^k\}^L(\lambda+\d)^{-1}w^i_nu^n_x\notag
  \\
  &\quad+\{w^i_nu^n_x{}_{\lambda+\mu+\d}u^k\}^L(\cancel{\lambda}
                 -\cancel{\lambda}-\mu-\d)^{-1}w^j_mu^m_x\notag
  \\
  &=g^{kl}(\lambda+\mu+\d)\d_lw^j_mu^m_x(\lambda+\d)^{-1}w^i_nu^n_x\notag
  \\
  &\quad+\Gamma^{kl}_su^s_x\d_lw^j_tu^t_x(\lambda+\d)^{-1}w^i_nu^n_x\notag
  \\
  &\quad-g^{kl}(\lambda+\mu+\d)^2w^j_l(\lambda+\d)^{-1}w^i_nu^n_x\notag
  \\
  &\quad-\Gamma^{kl}_su^s_x(\lambda+\mu+\d)w^j_l(\lambda+\d)^{-1}w^i_nu^n_x\notag
  \\
               &-\Big(i\leftrightarrow j,\lambda\leftrightarrow \mu\Big)\notag
  \\
%  \begin{split}
  &=g^{kl}\d_lw^j_su^s_xu^t_x-g^{kl}w^j_lw^i_su^s_x\lambda
                 -g^{kl}w^j_lw^i_su^s_x\mu \notag
  \\
  &\quad-g^{kl}\d(w^j_lw^i_su^s_x)-g^{kl}w^j_lw^i_su^s_x\mu
   -g^{kl}w^j_l\mu^2[(\lambda+\d)^{-1}w^i_nu^n_x] \notag
  \\
  &\quad-g^{kl}\d_sw^j_lu^s_x\mu[(\lambda+\d)^{-1}w^i_nu^n_x]
    -g^{kl}\d_sw^j_lw^i_tu^s_xu^t_x\notag
  \\
  &\quad-g^{kl}\d_sw^j_lu^s_x\mu[(\lambda+\d)^{-1}w^i_nu^n_x]
                 -g^{kl}\d^2w^j_l[(\lambda+\d)^{-1}w^i_nu^n_x]\notag
  \\
  &\quad-\Gamma^{kl}_sw^j_lw^i_tu^s_xu^t_x
    +g^{kl}\d_lw^j_su^s_x\mu[(\lambda+\d)^{-1}w^i_nu^n_x]
    \notag
  \\
  &\quad+g^{kl}\d(\d_lw^j_mu^m_x)[(\lambda+\d)^{-1}w^i_nu^n_x]
    +\Gamma^{kl}_s\d_lw^j_tu^s_xu^t_x[(\lambda+\d)^{-1}w^i_nu^n_x]
    \notag
  \\
  &\quad-\Gamma^{kl}_sw^j_lu^s_x\mu[(\lambda+\d)^{-1}w^i_nu^n_x]
    -\Gamma^{kl}_s\d_tw^j_lu^s_xu^t_x[(\lambda+\d)^{-1}w^i_nu^n_x]
    \notag
  \\
  &-\Big(i\leftrightarrow j,\lambda\leftrightarrow \mu\Big).
  %\end{split}
\end{align}
The notation $\big(i\leftrightarrow j,\lambda\leftrightarrow \mu\big)$ we used
in Equation \eqref{eq:NL3} means that all the terms are to be replaced with the
ones obtained by switching the corresponding indices and parameters. The
expressions of the form $[(\lambda+\d)^{-1}A]$ enclosed within square brackets
denote terms on which derivation operators ``from outside'' do not act and
containing derivations which do not act ``on the outside''.

The computation of the terms of $J^{ijk}_{\lambda,\mu}(L,N)$ is straightforward
for the first two addends
\begin{align}\label{eq:LN1}
%\begin{split}
\{u^i_{\lambda}&\{u^j_{\mu}u^k\}^L\}^N=\{u^i_\lambda
g^{kj}\}^N\mu+\{u^i_\lambda \Gamma^{kj}_su^s_x\}^N\notag
\\
&=\d_lg^{kj}w^l_su^s_x\mu[(\lambda+\d)^{-1}w^i_nu^n_x]+\d_l\Gamma^{kj}_su^s_xw^l_tu^t_x[(\lambda+\d)^{-1}w^i_nu^n_x]\notag
\\
&\quad+\Gamma^{kj}(\lambda+\d)w^l_su^s_x(\lambda+\d)^{-1}w^i_nu^n_x
\notag
\\
&=\d_lg^{kj}w^l_su^s_x\mu[(\lambda+\d)^{-1}w^i_nu^n_x]+\d_l\Gamma^{kj}_su^s_xw^l_tu^t_x[(\lambda+\d)^{-1}w^i_nu^n_x]
\notag
\\
&\quad+\Gamma^{kj}_lw^l_sw^i_tu^s_xu^t_x+
\Gamma^{kj}_l\d_tw^l_su^s_xu^t[(\lambda+\d)^{-1}w^i_nu^n_x]
\notag
\\
&\quad+\Gamma^{kj}_lw^l_su^s_{xx}[(\lambda+\d)^{-1}w^i_nu^n_x],
%\end{split}
\end{align}
\begin{align}\label{eq:LN2}
  \{u^j_{\mu}&\{u^i_{\lambda}u^k\}^L\}^N
  =\{u^j_\mu g^{ki}\}^N\lambda+\{u^j_\mu \Gamma^{ki}_su^s_x\}^N\notag
  \\
  &=\d_lg^{ki}w^l_su^s_x\lambda[(\mu+\d)^{-1}w^j_nu^n_x]
  +\d_l\Gamma^{ki}_su^s_xw^l_tu^t_x[(\mu+\d)^{-1}w^j_nu^n_x]
  \notag
  \\
  &\quad+\Gamma^{ki}_lw^l_sw^j_tu^s_xu^t_x+
  \Gamma^{ki}_l\d_tw^l_su^s_xu^t[(\mu+\d)^{-1}w^j_nu^n_x]
  \notag
  \\
&\quad+\Gamma^{ki}_lw^l_su^s_{xx}[(\mu+\d)^{-1}w^j_nu^n_x].
\end{align}
In the computation of the third one, we exploit the identity
$(A+\d)^{-1}B(A+\d)C=(A+\d)^{-1}(A+\d)BC-(A+\d)^{-1}[\d B]C$
\begin{align}\label{eq:LN3}
\{\{&u^i_\lambda
u^j\}^L_{\lambda+\mu}u^k\}^N=\{g^{ji}_{\lambda+\mu}u^k\}^N\lambda+\{\Gamma^{ji}_su^s_x{}_{\lambda+\mu}u^k\}
\notag
\\
&=w^k_nu^n_x(\lambda+\mu+\d)^{-1}w^l_su^s_x\d_lg^{ji}\lambda+w^k_nu^n_x(\lambda+\mu+\d)^{-1}w^l_su^s_x\d_l\Gamma^{ji}_tu^t_x
\notag
\\
&\quad-w^k_nu^n_x(\lambda+\mu+\d)^{-1}w^l_su^s_x(\lambda+\mu+\d)\Gamma^{ji}_l
\notag
\\
&=-\Gamma^{ji}_lw^k_sw^l_tu^s_xu^t_x
\notag
\\
&\quad+w^k_nu^n_x(\lambda+\mu+\d)^{-1}\left(\d_lg^{ji}w^l_su^s_x\lambda+\d_l\Gamma^{ji}_sw^l_tu^s_xu^t_x\right.
\notag
\\
&\qquad\left.+\d_sw^l_t\Gamma^{ji}_lu^t_xu^s_x+w^l_s\Gamma^{ji}_lu^s_{xx}\right).
\end{align}
We compute now the expression for $J^{ijk}_{\lambda,\mu}(N,N)$. We have
\begin{align}\label{eq:NN1}
  \{u^i_{\lambda}&\{u^j_{\mu}u^k\}^N\}^N=\{u^i_{\lambda}w^k_mu^m_x(\mu+\d)^{-1}w^j_nu^n_x\}^N
  \notag
  \\
&=w^k_mu^m_x(\lambda+\mu+\d)^{-1}\{u^i_\lambda
w^j_nu^n_x\}^N+[(\mu+\d)^{-1}w^j_nu^n_x]\{u^i_{\lambda}w^k_mu^m_x\}^N
\notag
\\
&=w^k_mu^m_x(\lambda+\mu+\d)^{-1}\left(\d_lw^j_su^s_xw^l_tu^t_x(\lambda+\d)^{-1}w^i_nu^n_x\right.
\notag
\\
&\qquad\left.+w^j_l(\lambda+\d)w^l_tu^t_x(\lambda+\d)^{-1}w^i_nu^n_x\right)
\notag
\\
&\quad+[(\mu+\d)^{-1}w^j_nu^n_x]\d_lw^k_su^s_xw^l_tu^t_x[(\lambda+\d)^{-1}w^i_nu^n_x]
\notag
\\
&\quad+[(\mu+\d)^{-1}w^j_nu^n_x]w^k_l(\lambda+\d)w^l_su^s_x(\lambda+\d)^{-1}w^i_nu^n_x
\notag
\\
&=w^k_mu^m_x(\lambda+\mu+\d)^{-1}\left(w^j_lw^l_sw^i_tu^t_xu^s_x\right)+
\notag
\\
&+\quad
w^k_mu^m_x(\lambda+\mu+\d)^{-1}\left(\left(\d_lw^j_sw^l_tu^s_xu^t_x+w^j_l\d_sw^l_tu^s_xu^t_x\right.\right.
\notag
\\
&\qquad\left.\left.+w^j_lw^l_su^s_{xx}\right)(\lambda+\d)^{-1}w^i_nu^n_x\right)+[(\mu+\d)^{-1}w^j_nu^n_x]w^k_lw^l_sw^i_tu^s_xu^t_x
\notag
\\
&\quad+[(\mu+\d)^{-1}w^j_nu^n_x][(\lambda+\d)^{-1}w^i_nu^n_x]\left(\d_lw^k_sw^l_tu^s_xu^t_x\right.
\notag
\\
&\qquad\left.+w^k_l\d_tw^l_su^s_xu^t+w^k_lw^l_su^s_{xx}\right)
\end{align}
\begin{align}\label{eq:NN2}
  \{u^j_{\mu}&\{u^i_{\lambda}u^k\}^N\}^N=w^k_mu^m_x(\lambda+\mu+\d)^{-1}\left(w^i_lw^l_sw^j_tu^t_xu^s_x\right)+
  \notag
  \\
&+\quad
w^k_mu^m_x(\lambda+\mu+\d)^{-1}\left(\left(\d_lw^i_sw^l_tu^s_xu^t_x+w^i_l\d_sw^l_tu^s_xu^t_x\right.\right.
\notag
\\
&\qquad\left.\left.+w^i_lw^l_su^s_{xx}\right)(\mu+\d)^{-1}w^j_nu^n_x\right)+[(\lambda+\d)^{-1}w^i_nu^n_x]w^k_lw^l_sw^j_tu^s_xu^t_x
\notag
\\
&\quad+[(\mu+\d)^{-1}w^j_nu^n_x][(\lambda+\d)^{-1}w^i_nu^n_x]\left(\d_lw^k_sw^l_tu^s_xu^t_x\right.
\notag
\\
&\qquad\left.+w^k_l\d_tw^l_su^s_xu^t+w^k_lw^l_su^s_{xx}\right)
\end{align}
\begin{align}\label{eq:NN3}
  \{\{u^i_{\mu}&u^j\}^N_{\lambda+\mu}u^k\}^N=\{w^j_mu^m_x(\lambda+\d)^{-1}w^i_nu^n_x{}_{\lambda+\mu}u^k\}^N\notag
  \\
  &=\{w^j_m u^m_x{}_{\lambda+\mu+\d}u^k\}^N(\lambda+\d)^{-1}w^i_nu^n_x\notag
  \\
  &\quad+\{(\lambda+\d)^{-1}w^i_n u^n_x{}_{\lambda+\mu+\d}u^k\}^Nw^j_mu^m_x
                 \notag
  \\
  &=\{w^j_m u^m_x{}_{\lambda+\mu+\d}u^k\}^N(\lambda+\d)^{-1}w^i_nu^n_x\notag
  \\
  &\quad-\{w^i_n u^n_x{}_{\lambda+\mu+\d}u^k\}^N(\mu+\d)^{-1}w^j_mu^m_x
                 \notag
  \\
  &=w^k_mu^m_x(\lambda+\mu+\d)^{-1}w^l_su^s_x\d_lw^j_tu^t_x
                 (\lambda+\d)^{-1}w^i_nu^n_x
                 \notag
  \\
  &\quad-w^k_mu^m_x(\lambda+\mu+\d)^{-1}w^l_su^s_x(\lambda+\mu+\d)
    w^j_l(\lambda+\d)^{-1}w^i_nu^n_t\notag
  \\
  &\quad-w^k_mu^m_x(\lambda+\mu+\d)^{-1}w^l_su^s_x\d_lw^i_tu^t_x
         (\mu+\d)^{-1}w^j_nu^n_x
                 \notag
  \\
  &\quad+w^k_mu^m_x(\lambda+\mu+\d)^{-1}w^l_su^s_x
    (\lambda+\mu+\d)w^i_l(\mu+\d)^{-1}w^j_nu^n_t
    \notag
  \\
  &=w^k_mu^m_x(\lambda+\mu+\d)^{-1}\left(\left(w^l_su^s_x\d_lw^j_tu^t_x
                 +\d_tw^l_su^s_xu^t_xw^j_l+\right.\right.
                 \notag
  \\
&\qquad\left.\left.+w^l_su^s_{xx}w^j_l
       \right)(\lambda+\d)^{-1}w^i_nu^n_x\right)
       \notag
  \\
  &\quad-w^k_mu^m_x(\lambda+\mu+\d)^{-1}\left(\left(w^l_su^s_x\d_lw^i_tu^t_x
                 +\d_tw^l_su^s_xu^t_xw^i_l+\right.\right.
                 \notag
  \\
               &\qquad\left.\left.+w^l_su^s_{xx}w^i_l \right)(\mu+\d)^{-1}w^j_nu^n_x\right)
                 \notag
  \\
               &\quad-w^k_sw^l_tw^j_lu^s_xu^t_x[(\lambda+\d)^{-1}w^i_nu^n_x]
                 \notag
  \\
&\quad+w^k_sw^l_tw^i_lu^s_xu^t_x[(\mu+\d)^{-1}w^j_nu^n_x].
\end{align}
Note that in the last passage we have used the same identity as in Equation
\eqref{eq:LN3} to simplify the terms of the form
$(\lambda+\mu+\d)^{-1} A(\lambda+\mu+\d)B$.

\subsection{Projection onto the basis}\label{ssec:basis}

The Jacobi identity lives in the previously defined space
$V_{\lambda,\mu}$. The partial results of our computation are not all expressed
in such a form that, after the expansion of the nonlocal terms, will produce
elements on the basis $\lambda^p\mu^{d-p}$ and
$(\lambda+\mu)^{-q}\lambda^{d+q}$ for $p\in\mathbb{Z}$, $q>0$ for all
$d\in\mathbb{Z}$.

All the double nonlocal terms in $J^{ijk}_{\lambda,\mu}(N,N)$ cancel out, leaving with a simplified expression
\begin{equation}
\label{eq:JNN}
\begin{split}
  J^{ijk}_{\lambda,\mu}(N,N)&=w^k_mu^m_x(\lambda+\mu+\d)^{-1}\left(w^j_lw^l_sw^i_tu^t_xu^s_x-w^i_lw^l_sw^j_tu^t_x\right)
  \\
  &\quad+[(\mu+\d)^{-1}w^j_nu^n_x]\left(w^k_lw^l_sw^i_tu^s_xu^t_x-w^k_sw^l_tw^i_lu^s_xu^t_x\right)
  \\
&\quad+[(\lambda+\d)^{-1}w^i_nu^n_x]\left(-w^k_lw^l_sw^j_tu^s_xu^t_x+w^k_sw^l_tw^j_lu^s_xu^t_x\right).
\end{split}
\end{equation}

There are terms which would not expand in the chosen basis in the last line of Equation \eqref{eq:NL2}. We have
\begin{align}\label{eq:norm1}
  w^k_m&u^m_x(\lambda+\mu+\d)^{-1}(g^{lj}\d_lw^i_su^s_x\mu)=
         \notag
  \\
  =&w^k_mu^m_x(\lambda+\mu+\d)^{-1}((\lambda+\mu+\d)g^{lj}\d_lw^i_su^s_x
  -(\lambda+\d)g^{lj}\d_lw^i_su^s_x)
  \notag\\
  =&g^{lj}\d_lw^i_sw^k_tu^s_xu^t_x-w^k_mu^m_x(\lambda
  +\mu+\d)^{-1}\big(g^{lj}\d_lw^i_su^s_x\lambda
    +\d_tg^{lj}\d_lw^i_su^s_xu^t_x
\notag \\
&+g^{lj}\d_{tl}w^i_su^s_xu^t_x+g^{lj}\d_lw^i_su^s_{xx}\big)
\end{align}
and
\begin{align}\label{eq:norm2}
  w^k_m&u^m_x(\lambda+\mu+\d)^{-1}w^i_l(\mu+\d)(g^{lj}\mu+\Gamma^{lj}_su^s_x)=
  \notag\\
  =&w^k_mu^m_x(\lambda+\mu+\d)^{-1}(\lambda+\mu+\d)w^i_l\left(g^{lj}\mu
    +\Gamma^{lj}_su^s_x\right)\notag
  \\
&-w^k_mu^m_x(\lambda+\mu+\d)^{-1}\left(\lambda w^i_l(g^{lj}\mu
  +\Gamma^{lj}_su^s_x)+\d_tw^i_lu^t_x(g^{lj}\mu+\Gamma^{lj}_su^s_x)\right)\notag
\\
=&g^{lj}w^i_lw^k_su^s_x\mu+\Gamma^{lj}_sw^k_tu^s_xu^t_x - w^k_mu^m_x(\lambda+\mu
+\d)^{-1}\left(w^i_l\Gamma^{lj}_su^s_z\lambda
  +\d_tw^i_l\Gamma^{lj}_su^s_xu^t_x\right)\notag
\\
&-w^k_mu^m_x(\lambda+\mu+\d)^{-1}\left((\lambda+\mu+\d)(g^{lj}w^i_l\lambda
  +g^{lj}\d_tw^i_lu^t_x)\right.\notag
\\
  &\qquad\left.-(\lambda+\d)(w^i_lg^{lj}\lambda+g^{lj}\d_tw^i_lu^t_x)\right)
         \notag
  \\
  =&g^{lj}w^i_lw^k_su^s_x\mu+\Gamma^{lj}_sw^k_tu^s_xu^t_x-g^{lj}w^i_lw^k_su^s_x\lambda-g^{lj}\d_sw^i_lw^k_tu^s_xu^t_x
     \notag
  \\
       &+w^k_mu^m_x(\lambda+\mu+\d)^{-1}\left(g^{lj}w^i_l\lambda^2+\d_sg^{lj}w^i_lu^s_x\lambda+2g^{lj}\d_sw^i_lu^s_x\lambda\right.
         \notag
  \\
       &\qquad\left.-\Gamma^{lj}_sw^i_lu^s_x\lambda+g^{lj}\d_{st}w^i_lu^s_xu^t_x+\d_sg^{lj}\d_tw^i_lu^s_xu^t_x-\Gamma^{lj}_s\d_tw^i_lu^s_xu^t_x\right.
         \notag
  \\
&\qquad\left.+g^{lj}\d_sw^i_lu^s_{xx}\right).
\end{align}

The full PVA-Jacobi identity can be then be obtained in terms of the previously
computed expressions, provided that we replace the last line in Equation
\eqref{eq:NL2} with the expression \eqref{eq:norm1} $+$ \eqref{eq:norm2}. The
full form of $J^{ijk}_{\lambda,\mu}$ is then
\begin{multline*}
  J^{ijk}_{\lambda,\mu}(P,P)=J^{ijk}_{\lambda,\mu}(L,L)
  +J^{ijk}_{\lambda,\mu}(L,N)+J^{ijk}_{\lambda,\mu}(N,L)+J^{ijk}_{\lambda,\mu}(N,N)
  \\
=\left(\text{\eqref{eq:LL1}} - \text{\eqref{eq:LL2}} -
  \text{\eqref{eq:LL3}}\right)
\\
+\left(\text{\eqref{eq:LN1}} - \text{\eqref{eq:LN2}} -
  \text{\eqref{eq:LN3}}\right)+\left(\text{\eqref{eq:NL1}} -
  \text{\eqref{eq:NL2}} - \text{\eqref{eq:NL3}}\right)
\\
+\left(\text{\eqref{eq:NN1}} - \text{\eqref{eq:NN2}}
  - \text{\eqref{eq:NN3}}\right).
\end{multline*}

\subsection{The conditions}

Assuming the skewsymmetry of the bracket $P$, the PVA-Jacobi equation is
symmetric for cyclic permutations of
$(i,\lambda),(j,\mu),(k,\nu=-\lambda-\mu-\d)$, and it is fulfilled if and only
if all the coefficients in the basis of $V_{\lambda,\mu}$ we have chosen
vanish. We report here the coefficients corresponding to the elements of
Section \ref{sec:conditions-1}, under the condition of skewsymmetry for the
bracket.
\begin{itemize}
\item The coefficient of $\lambda^2$ is
  \begin{equation}
    g^{li}\Gamma^{kj}_l-g^{kl}g^{ji}_{,l}+g^{kl}\Gamma^{ji}_l
  \end{equation}
  whose vanishing, given the skewsymmetry of the bracket, is equivalent to
  $g^{il}\Gamma^{kj}_l=g^{kl}\Gamma^{ij}_l$.
\item The coefficient of $\lambda\mu$ would be
  \begin{equation}
    g^{il}g^{kj}_{,l}-g^{lj}g^{ki}_{,l}-g^{kl}g^{ji}_{,l}+2g^{kl}\Gamma^{ji}_l
  \end{equation}
  which vanishes on account of the previous condition.
\item The coefficient of $\left((\mu+\d)^{-1}w^j_nu^n_x\right)\lambda^2$ is
  \begin{equation}
    g^{il}w^k_l-g^{kl}w^i_l
  \end{equation}
\item The expression that is obtained when $\lambda^0\mu^0=1$ is a differential
  polynomial. The coefficient multiplying $u^s_{xx}$ is
  \begin{equation}
    g^{kl}\left(\d_s\Gamma^{ji}_l-\d_l\Gamma^{ji}_s\right)
    +\Gamma^{kj}_l\Gamma^{li}_s-\Gamma^{ki}_l\Gamma^{lj}_s
    +g^{kl}\left(w^j_lw^i_s-w^i_lw^j_s\right),
  \end{equation}
  where the two summands with $ww$ come from \eqref{eq:NL3}.
\item The expression $\left((\mu+\d)^{-1}w^j_nu^n_x\right)$ is a differential
  polynomial. The coefficient multiplying $u^s_{xx}$ is
  \begin{equation}
    w^k_l\Gamma^{li}_s-w^l_s\Gamma^{ki}_l
    +g^{kl}\left(\d_lw^i_s-\d_sw^i_l\right)
  \end{equation}
  which is equal to Equation \eqref{eq:51} (after the exchange of the free
  indices $(j,m)\leftrightarrow (k,s)$).
\end{itemize}

\section{Concluding remarks}

In this paper we have considered three different approaches to the problem of
verifying Jacobi identity for weakly non local Poisson brackets of hydrodynamic
type, and we have showed their equivalence. While the equivalence between the
approach based on distributions and the approach based on (pseudo)-differential
operators is quite straightforward, the equivalence between the first two
approaches and the approach based on Poisson vertex algebras is more subtle and
requires additional work (Proposition 2 and Theorem 3). The final result is an
algorithmic procedure adapted to the three different formalisms. It is also
clear that the procedure can be used to check the compatibility of two Poisson
brackets, or that the Schouten bracket between two distinct Hamiltonian
operators vanish.

One of the most important consequences of the existence of a well-defined
procedure for the computation of the Jacobi identity is that the procedure can
be easily programmed in a computer algebra system.  We plan to do this in a
future work.

The extension of this work to weakly nonlocal symplectic operators \cite{Mal04}
seems possible along the lines presented in this paper, and might prove to be
useful. More generally, a differential-geometric theory of nonlocal
integrability operators (i.e., Hamiltonian operators, symplectic operators and
recursion operators for symmetries and conserved quantities) for partial
differential equations, including weakly nonlocal operators, already
exist~\cite{KVV18}, but it does not include Schouten brackets, the formulation
of the symplectic property and the formulation of the hereditary property for
recursion operators (the variational Nijenhuis bracket). However, this seems to
be a possible goal and at the moment it is in development, see \cite{KV18} for
latest advances.

\section{Appendix: three different recipes for the Jacobi
  identity}
\label{sec:appendix}

In this section we will show that the expression of the Jacobi identity can be
written in three different ways up to total divergencies.

We recall that the formal adjoint is defined by the equality
\begin{equation}
  \label{eq:53}
  [\langle A(\psi^1_i),\psi^2\rangle -   \langle \psi^1_i,A^*(\psi^2)\rangle]=0,
\end{equation}
where $\langle,\rangle$ is the pairing between vectors and covectors and square
brackets mean that the result is an equivalence class up to total
divergencies. In what follows we will need the standard facts \cite{Many,IVV}:
\begin{equation}
  \label{eq:25}
  \ell_{\Delta(\psi)}(\varphi) = \ell_{\Delta,\psi}(\varphi)
  + \Delta\circ\ell_{\psi}(\varphi),
\end{equation}
and
\begin{equation}
  \label{eq:24}
  \mathcal{E}(\langle\psi,\varphi\rangle) = \ell^*_\psi(\varphi) +
  \ell^*_\varphi(\psi),
\end{equation}
where $\mathcal{E}$ is the Euler--Lagrange operator.
If $P^*=-P$, then it is easy to prove \cite{IVV} that
\begin{equation}
  \label{eq:232}
  \ell^*_{P,\psi^1}(\psi^2) = \ell^*_{P^*,\psi^2}(\psi^1) =
  - \ell^*_{P,\psi^2}(\psi^1).
\end{equation}

\begin{theorem} Let $P$, $Q$ be skew-adjoint variational bivectors. Then,
  the following formulae for the Schouten bracket coincide up to total
  divergencies:
  \begin{align}
  \label{eq:22-b}
  \begin{split}
    [P,Q](\psi^1,\psi^2,\psi^3) = &
    \langle\ell_{P,\psi^1}(Q(\psi^2)),\psi^3\rangle +
    \text{cyclic}(\psi^1,\psi^2,\psi^3)
    \\
    & + \langle\ell_{Q,\psi^1}(P(\psi^2)),\psi^3\rangle +
    \text{cyclic}(\psi^1,\psi^2,\psi^3)
  \end{split}
  \\
  \label{eq:18}
  \begin{split}
    [P,Q](\psi^1,\psi^2,\psi^3) = &
    \langle\ell_{P,\psi^1}(Q(\psi^2)),\psi^3\rangle
    - \langle\ell_{P,\psi^2}(Q(\psi^1)),\psi^3\rangle
    \\
    &+ \langle\ell_{Q,\psi^1}(P(\psi^2)),\psi^3\rangle
    - \langle\ell_{Q,\psi^2}(P(\psi^1)),\psi^3\rangle
    \\
    & - \langle P(\ell^*_{Q,\psi^2}(\psi^1)),\psi^3\rangle
    - \langle Q(\ell^*_{P,\psi^2}(\psi^1)),\psi^3\rangle
  \end{split}
    \\
  \label{eq:21-b}
  \begin{split}
    [P,Q](\psi^1,\psi^2,\psi^3) = &
    \langle P(\mathcal{E}(\langle Q(\psi^1),\psi^2\rangle),\psi^3\rangle +
    \text{cyclic}(\psi^1,\psi^2,\psi^3)
    \\
    & + \langle Q(\mathcal{E}(\langle P(\psi^1),\psi^2\rangle),\psi^3\rangle +
    \text{cyclic}(\psi^1,\psi^2,\psi^3)
  \end{split}
\end{align}
\end{theorem}
\begin{proof}
The equivalence between \eqref{eq:22-b} and \eqref{eq:18} is given by the
following formulae (all equalities are up to total divergencies!):
\allowdisplaybreaks
\begin{align*}
  \langle\ell_{P,\psi^2}(Q(\psi^1)),\psi^3\rangle =&
       \langle Q(\psi^1),\ell^*_{P,\psi^2}(\psi^3)\rangle
  \\
  = &  \langle Q(\psi^1), \ell^*_{P^*,\psi^3}(\psi^2)\rangle
  \\
  = &  - \langle Q(\psi^1), \ell^*_{P,\psi^3}(\psi^2)\rangle
  \\
  = &  - \langle \ell_{P,\psi^3}(Q(\psi^1)), \psi^2\rangle
  \\
  - \langle Q(\ell^*_{P,\psi^2}(\psi^1)),\psi^3\rangle =&
     - \langle \ell^*_{P,\psi^2}(\psi^1),Q^*(\psi^3)\rangle
  \\
  =& \langle \ell^*_{P,\psi^2}(\psi^1),Q(\psi^3)\rangle
  \\
  =& \langle \psi^1,\ell_{P,\psi^2}(Q(\psi^3))\rangle
\end{align*}

The equivalence between \eqref{eq:22-b} and \eqref{eq:21-b} is given by the
following formulae:
\begin{align*}
  \langle P(\mathcal{E}(&Q(\psi^1)(\psi^2)),\psi^3\rangle = 
  - \langle \mathcal{E}(Q(\psi^1)(\psi^2)),P(\psi^3)\rangle
  \\
  =&  - \langle \ell^*_{Q(\psi^1)}(\psi^2)) + \ell^*_{\psi^2}(Q(\psi^1)),
      P(\psi^3)\rangle
  \\
  = & - \langle\psi^2,\ell_{Q,\psi^1}(P(\psi^3))\rangle
    - \langle\psi^2,Q(\ell_{\psi^1}(P(\psi^3)))\rangle
      + \langle\psi^1,Q(\ell_{\psi^2}(P(\psi^3)))\rangle
  \\
  = & \langle\psi^1,\ell_{Q,\psi^2}(P(\psi^3))\rangle
     - \langle\psi^2,Q(\ell_{\psi^1}(P(\psi^3)))\rangle
      + \langle\psi^1,Q(\ell_{\psi^2}(P(\psi^3)))\rangle
  \\
  = & \langle\psi^1,\ell_{Q,\psi^2}(P(\psi^3))\rangle
  +  \langle Q(\psi^2),\ell_{\psi^1}(P(\psi^3))\rangle
      - \langle Q(\psi^1),\ell_{\psi^2}(P(\psi^3))\rangle
\end{align*}
In the cyclic sum in \eqref{eq:21-b} all terms of the form
$ \langle Q(\psi^i),\ell_{\psi^j}(P(\psi^k))$ cancel if the computation is
restricted to covector-valued densities $\psi$ that lie in the image of the
Euler--Lagrange operator $\mathcal{E}$: $\psi = \mathcal{E}(F)$, where
$F = \int f\, dx$. In that case, we have $\ell_\psi = \ell^*_\psi$. The proof is
completed by the remark that multivector identities hold true in general even
if they are proved on the image of $\mathcal{E}$ only \cite{IVV}.
\end{proof}

\begin{remark}
  The expression~\eqref{eq:21-b} was used in \cite{Fer91b} in order to check
  the Jacobi identity. The expressions~\eqref{eq:22-b} and \eqref{eq:18}
  are more commonly used (see e.g. \cite{Many,Dorf,Olv93}). In particular, the
  expression~\eqref{eq:22-b} is the formula that we use throughout this
  paper. The three expressions do not exhaust all possibilities; see the above
  references for more exotic expressions of the Jacobi identity.
\end{remark}

\end{document}